\documentclass[a4paper,USenglish]{lipics-v2018}

\usepackage{microtype}
\usepackage{enumitem,amsmath,amssymb}
\usepackage{tikz}

\title{Random Subgroups of Rationals}

\author{Ziyuan Gao}{Department of Mathematics, National University of Singapore, Singapore}{matgaoz@nus.edu.sg}{}{}

\author{Sanjay Jain}{School of Computing, National University of Singapore, Singapore}{sanjay@comp.nus.edu.sg}{}{}

\author{Bakhadyr Khoussainov}{Department of Computer Science,
University of Auckland, New Zealand}{bmk@cs.auckland.ac.nz}{}{}

\author{Wei Li}{Department of Mathematics, National University of Singapore, Singapore}{matliw@nus.edu.sg
}{}{}

\author{Alexander Melnikov}{Institute of Natural and Mathematical Sciences, Massey University, New Zealand}{A.Melnikov@massey.ac.nz}{}{}

\author{Karen Seidel}{Hasso Plattner Institute, University of Potsdam, Germany}{karen.seidel@hpi.uni-potsdam.de
}{}{}

\author{Frank Stephan}{Department of Mathematics, National University of Singapore, Singapore}{fstephan@comp.nus.edu.sg
}{}{}

\authorrunning{Z.~Gao, S.~Jain, B.~Khoussainov, W.~Li, A.~Melnikov, K.~Seidel, F.~Stephan}

\Copyright{Ziyuan Gao, Sanjay Jain, Bakhadyr Khoussainov, Wei Li, Alexander Melnikov, Karen Seidel, and Frank Stephan}

\subjclass{}


\funding{Sanjay Jain was supported in part by NUS grant C252-000-087-001.
Furthermore, Ziyuan Gao, Sanjay Jain and Frank Stephan have been
supported in part by the Singapore Ministry of Education Academic Research Fund Tier 2
grant MOE2016-T2-1-019 / R146-000-234-112.
Bakhadyr Khoussainov was supported by the Marsden fund of Royal Society of New Zealand.
Karen Seidel was supported by the German Research Foundation (DFG) under Grant KO 4635/1-1 (SCL) and by the Marsden fund of Royal Society of New Zealand. }

\acknowledgements{The authors would like to thank Philipp Schlicht and Tin Lok Wong for helpful discussions, as well as 
thank Timo K\"{o}tzing for pointers to the literature.}

\nolinenumbers 

\EventEditors{} 
\EventNoEds{3}
\EventLongTitle{} 
\EventShortTitle{} 
\EventAcronym{} 
\EventYear{2019}
\EventDate{} 
\EventLocation{} 
\EventLogo{}
\SeriesVolume{}
\ArticleNo{}

\newcommand*{\Txt}{\mathbf{Txt}}
\newcommand*{\Bc}{\mathbf{Bc}}
\newcommand*{\Ex}{\mathbf{Ex}}
\newcommand{\la}{\lambda}

\newcommand{\res}{\!\upharpoonright\!}
\newcommand{\ran}{\mathrm{content}}

\newcommand{\theory}{\mathrm{Th}}

\newcommand*{\todo}[1]{\textcolor{red}{[\textbf{TODO:} #1]}}
\newcommand{\ks}[1]{\marginpar{\tiny{\textcolor{cyan}{K2S:} #1}}}

\newcommand{\N}{\mathbb{N}}
\newcommand{\Z}{\mathbb{Z}}
\newcommand{\sm}{\setminus}
\newcommand{\cF}{\mathcal{F}}
\newcommand{\cG}{\mathcal{G}}
\newcommand{\mg}{\mbox{mg}}
\newcommand{\NE}{\mbox{NE}}
\newcommand{\ve}{\varepsilon}
\newcommand{\T}{\mathcal{T}}
\newcommand{\cE}{\mathcal{E}}
\newcommand{\Lev}{\mathrm{Lev}}
\newcommand{\conv}{\!\downarrow}

\newcommand\spn[1]{{\left\langle#1\right\rangle}}

\theoremstyle{plain}
\newtheorem{proposition}[theorem]{Proposition}
\newtheorem*{claim}{Claim}
\newtheorem{nota}[theorem]{Notation}

\begin{document}

\maketitle

\begin{abstract}
This paper introduces and studies a notion of \emph{algorithmic randomness} for subgroups of rationals.
Given a randomly generated additive subgroup $(G,+)$ of rationals, two main questions are addressed: first, what are
the model-theoretic and recursion-theoretic properties of $(G,+)$; second, what learnability properties can one extract
from $G$ and its subclass of finitely generated subgroups?   
For the first question, it is shown that the theory of $(G,+)$ coincides with that of
the additive group of integers and is therefore decidable; furthermore, while the word problem for $G$ with respect 
to any generating sequence for $G$ is not even semi-decidable,
one can build a generating sequence $\beta$ 
such that the word problem for $G$ with respect to $\beta$ is co-recursively enumerable
(assuming that the set of generators of $G$ is limit-recursive).
In regard to the second question, it is proven that 
there is a generating sequence $\beta$ for $G$ such that every non-trivial finitely generated subgroup
of $G$ is recursively enumerable and the class of all such subgroups of $G$ is behaviourally correctly learnable,
that is, every non-trivial finitely generated subgroup can be semantically identified in the limit
(again assuming that the set of generators of $G$ is limit-recursive).
On the other hand, the class of non-trivial finitely generated subgroups of $G$ cannot be
syntactically identified in the limit with respect to any generating sequence for $G$.
The present work thus contributes to a recent line of research studying algorithmically random
infinite structures and uncovers an interesting connection between the arithmetical complexity of 
the set of generators of a randomly generated subgroup of rationals and the learnability of its
finitely generated subgroups.     
\end{abstract}

\newpage
\section{Introduction}
\label{sec:intro}

The concept of \emph{algorithmic randomness}, particularly for strings and infinite sequences,
has been extensively studied in recursion theory and theoretical computer science \cite{dowhir10,LiV08,nies09}.
It has also been applied in a wide variety of disciplines, including formal language and automata
theory \cite{LiV95}, machine learning \cite{VovkGS99}, and recently even quantum theory \cite{NiesS18}. 
An interesting and long open question is whether the well-established notions of randomness for infinite sequences have analogues
for infinite structures such as graphs and groups.  Intuitively, it might be reasonable
to expect that a collection of random infinite structures possesses the following characteristics:
(1) randomness should be an isomorphism invariant property; in particular, random structures
should not be computable;
(2) the collection of random structures (of any type of algebraic structure) should have cardinality equal 
to that of the continuum.
The standard random infinite graph thus does not qualify as an algorithmically random structure;
in particular, it is isomorphic to a computable graph and has a countable categorical theory.
Very recently, Khoussainov \cite{Kho14,Kho15} defined algorithmic randomness for infinite structures 
that are akin to graphs, trees and finitely generated structures.

This paper addresses the following three open questions in algorithmic
randomness: (A) is there a reasonable way to define algorithmically random structures for standard
algebraic structures such as groups; (B) can one define algorithmically randomness for
groups that are not necessarily finitely generated; 
(C) what are the model-theoretic properties of algorithmically random structures?
The main contribution of the present paper is to answer these three questions positively for a 
fundamental and familiar algebraic structure, the \emph{additive group of rationals}, 
denoted $(\mathbb{Q},+)$. 
Prior to this work, question (A) was answered for structures such as
\emph{finitely generated} universal algebras, connected graphs, trees of bounded degree and
monoids \cite{Kho14}.  
Concerning question (C), it is still unknown whether the first order theory of
algorithmically random graphs (or trees) is decidable. In fact, it is not even
known whether any two algorithmically random graphs (of the same bounded degree)
are elementarily equivalent~\cite{Kho14}.

As mentioned earlier, one goal of this work is to formulate a notion of randomness for 
subgroups of $(\mathbb{Q},+)$.  This is a fairly natural class of groups to consider, 
given that the isomorphism types of its subgroups have been 
completely classified, as opposed to the current limited state of knowledge about the isomorphism 
types of even rank $2$ groups.  As has been known since the work of Baer \cite{Bae37},
the subgroups of $(\mathbb{Q},+)$ coincide, up to isomorphism, with the torsion-free Abelian groups of rank $1$.
Moreover, the group $(\mathbb{Q},+)$ is robust enough that it has uncountably many algorithmically 
random subgroups (according to our definition of algorithmically random subgroups of $(\mathbb{Q},+)$), which
contrasts with the fact that there is a unique standard random graph up to isomorphism.
At the same time, the algorithmically random subgroups of
$(\mathbb{Q},+)$ are not too different from one other in the sense that they are all elementarily equivalent
(a fact that will be proven later), 
which is similar to the case of standard random graphs being elementarily equivalent.
    
The properties of the subgroups of $(\mathbb{Q},+)$ were first systematically studied by 
Baer \cite{Bae37} and then later by Beaumont and Zuckerman \cite{BeaZ51}. 
Later, the group $(\mathbb{Q},+)$
was studied in the context of automatic structures \cite{Tsa11}.
An early definition of a random group is due to Gromov \cite{Gro03}.  
According to this definition, random groups are those 
obtained by first fixing a set of generators, and then randomly choosing 
(according to some probability distribution) the relators specifying the 
quotient group.
An alternative definition of a general random infinite structure was proposed by Khoussainov \cite{Kho14,Kho15};
this definition is based on the notion of a \emph{branching class}, which is in turn used to define Martin-L\"{o}f
tests for infinite structures entirely in analogy to the definition of a Martin-L\"{o}f test for sequences.
An infinite structure is then said to be Martin-L\"{o}f random if it passes every
Martin-L\"{o}f test in the preceding sense.
The existence of a branching class of groups, and thus of continuunm many Martin-L\"{o}f random groups,
was only recently established \cite{HaKT18}.    
 
Like Gromov's definition of a random group, the one adopted in the present work is syntactic, in 
contrast to the semantic and algebraic definition due to Khoussainov.  However, rather than selecting 
the relators at random according to a prescribed probability distribution for a fixed set of generators, 
our approach is to directly encode a Martin-L\"{o}f random binary sequence into the generators of the subgroup.    
More specifically, we fix any binary sequence $R$, and define the group
$G_R$ as that generated by all rationals of the shape $p_i^{-n_i}$, where $p_i$ denotes
the $(i+1)$-st prime and $n_i$ is the number of ones occurring between the $i$-th and
$(i+1)$-st occurrences of zero in $R$; $n_0$ is the number of starting ones, and if there is no 
$(i+1)$-st zero then $n_j$ is 
defined to be zero for all $j$ greater than $i$ and $G_R$ is generated by all $p_{i'}^{-n_{i'}}$
with $i'$ less than $i$ and all $p_i^{-n'}$ such that $n'$ is any positive integer.
$G_R$ is then said to be \emph{randomly generated} if and only if $R$ is Martin-L\"{o}f
random.  In order to derive certain computability properties, it will always be assumed
in the present paper that any Martin-L\"{o}f random sequence associated to a 
randomly generated subgroup of $(\mathbb{Q},+)$ is also limit-recursive.  
It may be observed that no finitely generated subgroup of $(\mathbb{Q},+)$
is randomly generated in the sense adopted here; this corresponds to the intuition that
in any ``random'' infinite binary sequence $R$, the fraction of zeroes in the first $n$ bits 
should tend to a number strictly smaller than one as $n$ grows to infinity. 
For a similar reason, no randomly generated subgroup $G_R$ 
is infinitely divisible by a prime, that is, there is no prime $p$ such that $p^{-n}$ belongs to $G_R$
for all $n$.  

The first main part of this work is devoted to the study of the model-theoretic and 
recursion-theoretic properties of randomly generated subgroups of $(\mathbb{Q},+)$.
It is shown that the theory of any randomly generated subgroup 
coincides with that of the integers with addition (denoted $(\Z,+)$), and is therefore
decidable\footnote{For a proof of the decidability of the theory of $(\Z,+)$, often 
known as \emph{Presburger Arithmetic}, see \cite[pages 81--84]{Mar02}.}.
Next, we define 
the notion of a \emph{generating sequence} for a randomly generated group $G_R$;
this is an infinite sequence $\beta$ such that $G_R$ is generated by the terms of $\beta$.
We then consider the word problem for $G_R$ with respect to $\beta$: in detail,
this is the problem of determining, given any two finite integer sequence representations 
$\sigma$ and $\tau$ of elements of $G_R$
with respect to $\beta$, whether or not $\sigma$ and $\tau$ represent the same
element of $G_R$.  We show that 
the word problem for $G_R$ with respect to \emph{any} generating sequence $\beta$
is never recursively enumerable (r.e.); on the other hand, one can construct a generating
sequence $\beta'$ for $G_R$ such that the corresponding word problem
for $G_R$ is co-r.e.  Moreover, one can build a generating
sequence $\beta''$ for $G_R$ such that the word problem for the quotient group
of $G_R$ by $\Z$ with respect to $\beta''$ is r.e. 

The second main part of this paper investigates the learnability of non-trivial finitely generated subgroups
of randomly generated subgroups of $(\mathbb{Q},+)$ from positive examples, also
known as learning from text.  
Stephan and Ventsov \cite{SteV01} examined the 
learnability of classes of substructures of algebraic structures; the study of
more general classes of structures was undertaken in the work of Martin and Osherson \cite[Chapter III]{MarO98}.
The general objective is to understand how semantic knowledge of a class of concepts
can be exploited to learn the class; in the context of the present problem, semantic
knowledge refers to the properties of every finitely generated subgroup of
any randomly generated subgroup of rationals, such as being generated by a single
rational \cite{Bae37}.  It may be noted that the present work considers learning
of the actual representations of finitely generated subgroups, which are all
isomorphic to each other, as opposed to learning their structures up to
isomorphism, as is considered in the learning framework of Martin and Osherson \cite{MarO98}.
Various positive learnability results are obtained: it will be proven, for example, that
for any randomly generated subgroup $G_R$ of $(\mathbb{Q},+)$, there is a generating sequence
$\beta$ for $G_R$ such that the set of representations of every non-trivial finitely generated 
subgroup of $G_R$
with respect to $\beta$ is r.e.; furthermore, the class of all such representations is
\emph{behaviourally correctly} learnable, that is, all these representations can be identified
in the limit up to semantic equivalence.  On the other hand, it will be seen
that the class of all such representations can never be \emph{explanatorily} learnable, or learnable 
in the limit.  Similar results hold for the class of non-trivial finitely
generated subgroups of the quotient group of $G_R$ by $\Z$.  Thus this facet of our work implies a 
connection between the limit-recursiveness
of the set of generators of a randomly generated subgroup of $(\mathbb{Q},+)$ and the learnability
of its non-trivial finitely generated subgroups. 

\section{Preliminaries}\label{sec:not}

\noindent
Any unexplained recursion-theoretic notation may be found in \cite{Rog67,Soa87,Od89}.
For background on algorithmic randomness, we refer the reader to \cite{dowhir10,nies09}.
We use $\N =\{0,1,2,\ldots\}$ to denote
the set of all natural numbers and $\Z$ to denote the set of all integers.
The $(i+1)$-st prime will be denoted by $p_i$.
$\Z^{<\omega}$ denotes the set of all finite sequences of integers.
Throughout this paper, $\varphi_0,\varphi_1, \varphi_2,\ldots$ is a fixed 
acceptable programming system of all partial recursive functions and 
$W_0,W_1,W_2,\ldots$ is a fixed \emph{acceptable numbering of all recursively 
enumerable} (abbr.~r.e.) \emph{sets} of natural numbers.
We will occasionally work with objects belonging to some countable class $X$ different from $\N$; in 
such a case, by abuse of notation, we will use the same symbol $W_e$ to denote
the set of objects obtained from $W_e$ by replacing each member $x$ with
$F(x)$ for some fixed bijection $F$ between $\N$ and $X$. 
   

Given any set $S$, 
$S^*$ denotes the set of all finite sequences of elements from~$S$. By
$D_0,D_1,D_2,\ldots$ we denote any fixed \emph{canonical indexing of all finite sets}
of natural numbers. 
Cantor's pairing function $\langle \,\cdot\,,\,\cdot\,\rangle\colon \N\times\N\to\N$
is given by $\langle x,y\rangle =\frac{1}{2}(x+y)(x+y+1)+y$ 
for all $x,y\in\N$.
The symbol $K$ denotes the \emph{diagonal  
halting problem}, i.e.,  $K=\{e\mid e\in\N,\; \varphi_e(e)~\mbox{converges}\}$.
The \emph{jump} of $K$, that is, the relativised halting problem $\{e\mid e \in \N; \varphi_e^K(e)\conv\}$, 
will be denoted by $K'$.

For $\sigma\in (\N\cup \{\#\})^{\ast}$ and $n\in\N$
we write $\sigma(n)$ to denote the element in the $n$-th position of 
$\sigma$. Further, $\sigma[n]$ denotes the
sequence $\sigma(0),\sigma(1),\ldots,\sigma(n-1)$.
Given a number $a\in\N$ and some fixed $n\in\N$, 
$n\geq 1$, we denote by $a^n$ the finite sequence $a,\ldots, a$, 
where $a$ occurs exactly $n$ times.  Moreover, 
we identify $a^0$ with the empty string $\ve$.
For any finite sequence $\sigma$ we use
$|\sigma|$ to denote the length of $\sigma$.  
The concatenation of two sequences $\sigma$ and $\tau$ is
denoted by $\sigma\circ\tau$; for convenience, and whenever there is no
possibility of confusion, this is occasionally
denoted by $\sigma\tau$.  For any sequence $\beta$ (infinite or otherwise)
and $s < |\beta|$, $\beta\res_s$ denotes the initial segment of $\beta$
of length $s+1$.
For any $m \geq 1$ and $p \in \Z$,
$I_m(p)$ denotes the vector of length $m$ whose first $m-1$ coordinates
are $0$ and whose last coordinate is $p$.
Furthermore, given two vectors $\alpha = (a_i)_{0\leq i\leq m}$
and $\beta = (b_i)_{0 \leq i \leq m}$ of equal length,
$\alpha \cdot \beta$ denotes the scalar product of
$\alpha$ and $\beta$, that is, $\alpha \cdot \beta :=
\sum_{i=0}^m a_ib_i$.  For any $c \in \Z$ and $\sigma := (b_i)_{0 \leq i \leq m} \in \Z^{<\omega}$,
$c\sigma$ denotes the vector obtained from $\sigma$ by coordinatewise multiplication with $c$, that is, $c\sigma := (cb_0,cb_1,\ldots,cb_m)$. 
For any non-empty $S \subseteq \mathbb{Q}$, $\spn{S}$ denotes 
$\{\sum_{i=0}^k c_i s_i\mid k \in \N \wedge c_i \in \Z \wedge s_i \in S\}$. 

Cantor space, the set of all infinite binary sequences, will be denoted by $2^{\omega}$.
The set of finite binary strings will be denoted by $2^{< \omega}$.
For any binary string $\sigma$, $[\sigma]$ denotes the cylinder generated by
$\sigma$, that is, the set of infinite binary sequences with prefix $\sigma$.
For any $U \subseteq 2^{< \omega}$, the open set generated by $U$ is
$[U] := \bigcup_{\sigma \in U} [\sigma]$.  The Lebesgue measure on $2^{\omega}$
will be denoted by $\la$; that is, for any binary string $\sigma$,
$\la([\sigma]) = 2^{-|\sigma|}$.
By the Carathéodory Theorem, this uniquely determines the Lebesgue measure on the Cantor space.

\section{Randomly Generated Subgroups of Rationals}\label{sec:StringRandomGroups}
\label{sec:prelims}

We first review some basic definitions and facts in algorithmic randomness which in our setting is always understood w.r.t the Lebesgue measure.
An \emph{r.e.\ open set} $R$ is an open set generated by an r.e.\ set of binary strings. 
Regarding $W_e$ as a subset of $2^{<\omega}$, one has an enumeration
$[W_0],[W_1],[W_2],\ldots$ of all r.e.\ open sets.
A \emph{uniformly r.e.\ sequence $(G_m)_{m < \omega}$ of open sets}
is given by a recursive function $f$ such that $G_m = [W_{f(m)}]$ for
each $m$.  As infinite binary sequences may be viewed as characteristic
functions of subsets of $\N$, we will often use the term ``set''
interchangeably with ``infinite binary sequence''; in particular,
the subsequent definitions apply equally to subsets of $\N$ and
infinite binary sequences.  

Martin-L\"{o}f \cite{lof66} defined randomness based on tests.  A \emph{Martin-L\"{o}f 
test} is a uniformly r.e.\ sequence $(G_m)_{m < \omega}$ of open sets such that 
$(\forall m < \omega)[\la(G_m) \leq 2^{-m}]$.
A set $Z \subseteq \N$ \emph{fails} the test if $Z \in \bigcap_{m < \omega} G_m$;
otherwise $Z$ \emph{passes} the test.
$Z$ is \emph{Martin-L\"{o}f random} if $Z$ passes each Martin-L\"{o}f test.   

Schnorr \cite{schnorr71} showed that Martin-L\"{o}f random sets can be described via martingales.  A \emph{martingale} is a function $\mg:2^{<\omega} \rightarrow \mathbb{R}^+ \cup \{0\}$
that satisfies for every $\sigma \in 2^{<\omega}$ the equality $\mg(\sigma\circ 0) + \mg(\sigma\circ 1)= 2\mg(\sigma)$.
For a martingale $\mg$ and a set $Z$, the martingale $\mg$ \emph{succeeds} on $Z$ if $\sup_n \mg(Z(0)\ldots Z(n)) = \infty$.

\begin{theorem}\cite{schnorr71}\label{thm:martinlofrandommartingale}
For any set $Z$, $Z$ is Martin-L\"{o}f random iff no r.e.\ martingale succeeds on $Z$. 
\end{theorem}

\noindent
The following characterisation of all subgroups of $(\mathbb{Q},+)$ forms the basis of our definition of a random subgroup.

\begin{theorem}\cite{BeaZ51}\label{thm:subgroupsrationalschar}
Let $G$ be any subgroup of $(\mathbb{Q},+)$.  Then there is an integer $z$, as well as a sequence
$(n_i)_{i < \omega}$ with $n_ i \in \N \cup \{\infty\}$ such that $G = \left\{\displaystyle\frac{a \cdot z}{\Pi_{i=0}^k p_i^{m_i}}\mid a \in \Z \wedge k \in \N \right. \wedge(\forall i \leq k)[\left. m_i \in \N \wedge m_i < n_i] \right\}$.
\end{theorem}

\begin{definition}
\label{Group1}
Let $R \in 2^\omega$ be a real in the Cantor space, i.e. an infinite sequence of $0$'s and $1$'s.
Then the group $G_R$ is the subgroup of the rational numbers $(\mathbb{Q},+)$ generated by $a_0, a_1, \ldots$ with
$a_i=\frac{1}{p_i^{n_i}}$ for all $i \in \mathbb{N}$,
where for each $i \in \mathbb{N}$, by $p_i$ we denote the $(i+1)$-st prime and by $n_i$ the number of consecutive $1$'s in $R$ between the $i$-th and $(i+1)$-st zero in $R$, with which we let $n_0$ count the number of starting $1$'s.
If there is no $(i+1)$-st zero, we let $n_i := \infty$, meaning that for all $n$ the fraction $\frac{1}{p_i^n}$ is in $G_R$.
\end{definition}

Clearly, $(\mathbb{Z},+)$ is always a subgroup of $G_R$ and $\frac{1}{p_i} \notin G_R$ if and only if the $i$-th and $(i+1)$-st zero in $R$ are consecutive. Thus, if $R$ ends with infinitely many zeros, then $G_R$ is isomorphic to $(\mathbb{Z},+)$.
Moreover, there is a prime $p_i$ such that $\frac{1}{p_j} \notin G_R$ for all $j > i$ and $\frac{1}{p_i^n} \in G_R$ for all $n\in\mathbb{N}$, for short $p_i$ infinitely divides $G_R$, if and only if $R$ ends with an infinite sequence of $1$'s.

\begin{lemma}
If $R \in 2^\omega$ is Martin-Löf random
, then $n_i$ is finite for every $i \in \mathbb{N}$, where $n_i$ is defined as in Definition \ref{Group1}. 
In other words, the group $G_R$ is not infinitely divisible by any prime.
\end{lemma}
\begin{proof}
This is an easy observation, as in no Martin-Löf random w.r.t the Lebesgue measure only finitely many $0$'s occur.
\end{proof}

\noindent A similar argument shows that for Martin-Löf random $R$ there are infinitely many primes occurring as basis of a denominator of a generator.

\begin{definition}
Fix a probability distribution $\mu$ on the natural numbers and let $X=(X_i)_{i\in \mathbb{N}}$ be a sequence of iid random variables taking values in $\mathbb{N}$ with distribution $X_i \sim \mu$ for all $i\in\mathbb{N}$.
Denote by $H_X$ the subgroup of $(\mathbb{Q},+)$ generated by $\{ p_i^{-X_i} \mid i \in \mathbb{N} \}$, where $p_i$ denotes the $(i+1)$-st prime.
\end{definition}

The so obtained random group might follow a more uniform process.

\begin{lemma}
If $\mu$ is the distribution on $\mathbb{N}$ assigning $0$ probability $\frac{1}{2}$, $1$ probability $\frac{1}{4}$, $2$ probability $\frac{1}{8}$ and $n$ probability $2^{-n-1}$, then with probability $1$ holds $H_X=G_R$ for some Martin-Löf random $R$. 
\end{lemma}
\begin{proof}
This follows immediately, as the set of ML-randoms has measure $1$ with respect to the Lebesgue measure. From $X_0=n_0$, $X_1=n_1$, $\ldots$, $X_i=n_i$, $\ldots$ we obtain an infinite binary sequence $R\in 2^\omega$ by recursively appending $1^{n_i}0$ in step $i$ to the already established initial segment of $R$, starting with the empty string.
By definition the Lebesgue measure assigns probability $\frac{1}{2^{n+1}}$ to having the (intermediate) subsequence $1^{n}0$ in $R$. This is exactly the probability of the event $X_i=n$. 
\end{proof}

\noindent
A \emph{generating sequence for $G_R$} is an infinite sequence $(b_i)_{i<\omega}$ such that
$\spn{b_i\mid i < \omega} = G_R$.  We will often deal with generating sequences rather than minimal generating
sets for $G_R$, mainly due to the fact that if the terms of a sequence $\beta$ are carefully chosen based on a limiting 
recursive programme for $R$ (so that $\beta$ itself is limiting recursive), then, as will be seen later, the set of representations 
of elements of $G_R$ with respect to $\beta$ can have certain desirable computability properties, such as equality being
co-r.e.

\begin{proposition}\label{prop:recursivegeneratingsequence}
Suppose $R \leq_T K$ is Martin-Löf random. 
Then there does not exist any strictly increasing recursive enumeration $i_0,i_1,i_2,\ldots$
such that for each $j$, there is some $n_{i_j} \geq 1$ with $p_{i_j}^{-n_{i_j}} 
\in G_R$. 
\end{proposition}
\def\proofofrecursivegeneratingsequence{
\begin{proof} 
Suppose that such an enumeration did exist.  We show that this contradicts the Martin-L\"{o}f
randomness of $R$.  By Theorem \ref{thm:martinlofrandommartingale},
it suffices to show that there is a recursive martingale $\mg$ succeeding on $R$.  
Define $\mg$ as follows.
For any $\sigma \in \{0,1\}^*$, if there is some $j$ such
that $\sigma$ contains at least $i_j$ occurrences of $0$ and the $i_j$-th occurrence of $0$
is immediately succeeded by $0$, then set $\mg(\sigma) = 0$.  Else, let $j$ be the largest
$j'$ for which either $j' = 0$ or $\sigma$ contains at least $i_{j'}$ occurrences of $0$,    
and set 
$$
\mg(\sigma) = \left\{\begin{array}{ll}
2^{j+1} & \mbox{if $\sigma$ contains at least $i_j$ $0$'s and the $i_j$-th occurrence of $0$ in $\sigma$ is not} \\
~ & \mbox{the last bit of $\sigma$;} \\
2^j & \mbox{otherwise.}\end{array}\right.
$$  
It may be directly verified that $\mg$ satisfies the martingale equality 
$\mg(\sigma) = \frac{1}{2}(\mg(\sigma0)+\mg(\sigma1))$ for all $\sigma \in \{0,1\}^*$.
Furthermore, $\mg(R(0)R(1)\ldots R(n))$ grows to infinity with $n$ and so $\mg$ succeeds on $R$,
contradicting the fact that $R$ is Martin-Löf random.         
\end{proof}
}
\proofofrecursivegeneratingsequence

\begin{theorem}\label{thm:genseqcoreequality}
\label{coRE}
If $R \leq_T K$ is Martin-Löf random, 
then $(G_R,+)$ is co-r.e., meaning that $+$ is recursive and there is a generating sequence with respect to which equality is co-r.e.
\end{theorem}
\begin{proof}
For a fixed generating sequence $( q_i )_{i < \omega}$ of $G_R$ there is an epimorphism from the set of finite sequences of integers $\mathbb{Z}^{<\omega}$ to $G_R$ by identifying $\sigma=(\sigma(0),\ldots,\sigma(|\sigma|-1))$ with $x= \sum_{i=0}^{|\sigma|-1} \sigma(i) q_i$.
We call $\sigma$ a representation of $x$ w.r.t.~$(q_i)_{i<\omega}$ or $(q_i)_{i<n+1}$.

\smallskip
Obviously, for any generating sequence of $G_R$ addition is recursive as only the components of the representations have to be added as integers.

\smallskip
In order to prove that equality is co-r.e., we construct a specific generating sequence $(b_i)_{i<\omega}$.
Based on the result $R^s$ of the computation of $R$ after $s$ steps, we are going to define finite sequences $\beta_s$ of rational numbers recursively, such that $|\beta_s|=s+1$ and inequality on $\{-s-1,\ldots,s+1\}^{s+1} \subseteq \mathbb{Z}^{s+1}$, interpreted as representations w.r.t.~$\beta_s$, is decided and extends the inequalities on $\{-s,\ldots,s\}^{s}$, even though they originate from an interpretation as representations according to $\beta_{s-1}$.
With this in the limit we obtain a generating sequence of $G_R$, meaning that for every $i$ there is some $s_i > i$ such that for all $s \geq s_i$ the $i$-th element of $\beta_s$ is the same as the $i$-th element of $\beta_{s_i}$, which we denote by $b_i$. Further, $(b_i)_{i \in \mathbb{N}}$ generates $G_R$ and for this generating sequence equality will be co-r.e.

In the following we write $n_{i,s}$ for 
$n_i$ according to $R^s$, i.e. the number of $1$'s between the $i$-th and $(i+1)$-st zero in $R^s$, as introduced in Definition \ref{Group1}.
As $R^s$ does not end with infinitely many $1$'s, $n_{i,s}$ can be computed in finitely many steps for every $i$ and $s$.

\begin{enumerate}
\item[]$s=0$.
Let $\beta_0=(1)$.
\item[]$s\leadsto s+1$.
Check for every $i\leq s$ whether $n_{i,s}= n_{i,s+1}$.
If $n_{i,s}=n_{i,s+1}$ let $\beta_{s+1}(i)=\beta_s(i)$.
Replace all $\frac{1}{p_i^{n_{i,s}}}$ occurring in $\beta_{s}$ with $n_{i,s} \neq n_{i,s+1}$ by some respective integer, for which existence we argue below, such that
\begin{align*}
\Delta_{(q_i)_{i<s+1}} = \{ \,(\sigma_0, \sigma_1) \in \,&(\,\{-s-1,\ldots,s+1\}^{s+1}\,)^2 \mid \\
& \sigma_0, \sigma_1 \text{ represent different elements w.r.t. } (q_i)_{i<s+1} \,\}
\end{align*}
stays the same or enlarges if $(q_i)_{i<s+1}$ equals the first $s+1$ entries of $\beta_{s+1}$ instead of $\beta_s$.
Further, let
$$\beta_{s+1}(s+1) = \frac{1}{p_j^{n_{j,s+1}}},$$
where $j\leq s+1$ is minimal such that $\frac{1}{p_j^{n_{j,s+1}}}$ is an element of $G_{R^{s+1}}$ and does not yet occur in $\beta_{s+1}\res (s+1)$.
If there is no such $j$, let $\beta_{s+1}(s+1)=1$.
\end{enumerate}

For example, if the tape after stage $s=2$ started with $1111010\ldots$, after $3$ steps contained $1101010\ldots$ and $\beta_2=(1,\frac{1}{2^4},\frac{1}{3})$, then in $\beta_3$ we would have to replace $\frac{1}{2^4}$ by an integer $w$ such that for arbitrary integers $u_0,u_1,u_2,v_0,v_1,v_2$ between $-3$ and $3$ we have
$$u_0 + u_1 \frac{1}{2^4} + u_2 \frac{1}{3} \neq v_0 + v_1 \frac{1}{2^4} + v_2 \frac{1}{3} \quad \Rightarrow \quad u_0 + u_1 w + u_2 \frac{1}{3} \neq v_0 + v_1 w + v_2 \frac{1}{3}$$
and $\beta_{3}(3)$ would be $\frac{1}{2^2}$.


\medskip
We proceed by showing that there is always such an integer $w$.

\begin{claim}
For every $s\in\mathbb{N}$ in step $s+1$ it is possible to alter finitely many entries of $\beta_s$ to obtain $\beta_{s+1}\res(s+1)$ such that $\Delta_{\beta_s}\subseteq\Delta_{\beta_{s+1}\res(s+1)}$.
\end{claim}
\begin{proof}[Proof of the Claim.]
Let $s\in\mathbb{N}$. It suffices to show that one entry can be replaced in this desired way. As the argument does not depend on the position, we further assume that it is the last entry.
For all $(\sigma_0, \sigma_1) \in \Delta_{\beta_s}$ we want to prevent
$$\sum_{i=0}^{s-1} \sigma_0(i) \beta_{s}(i) + \sigma_0(s)w = \sum_{i=0}^{s-1} \sigma_1(i) \beta_{s}(i) + \sigma_1(s)w.$$
This is a linear equation having zero or one solution in $\mathbb{Q}$. As there are only finitely many choices for the pair $(\sigma_0, \sigma_1)$, an integer not fulfilling any of these equations can be found in a computable way.
\end{proof}

We continue by proving that the entries of the $\beta_s$ stabilize, such that in the limit we obtain a sequence $(b_i)_{i<\omega}$ of elements of $G_R$.

\begin{claim}
For every $i\in\mathbb{N}$ there is some $s_i\geq i$ such that for all $s\geq s_i$ we have $\beta_s(i)=b_i$, with $b_i=\beta_{s_i}(i)$.
\end{claim}
\begin{proof}[Proof of the Claim.]
Let $i \in \mathbb{N}$. If there is $s_i>i$ such that the entry $\beta_{s_i-1}(i)$ had to be changed, then $\beta_{s_i}(i)$ is an integer and thus, it will never be changed lateron.
In case this does not happen, we obtain $\beta_s(i)=\beta_i(i)$ for all $s \geq i$ and therefore $s_i=i$.
\end{proof}

By the next claim the just constructed sequence generates the random group.

\begin{claim}
The sequence $(b_i)_{i<\omega}$ generates $G_R$.
\end{claim}
\begin{proof}[Proof of the Claim.]
Let $i\in \mathbb{N}$ and $a_i$ as in Definition~\ref{Group1}. We argue that there is some $j$ with $a_i=b_j$. Let $m_i$ be the position of the $(i+1)$-st zero in the Martin-Löf random $R$. Then there is $s'$ such that after $s'$ computation steps $R\res (m_i+1)$ is not changed any more. Thus, after at most $i$ additional steps all generators of $G_R$ having one of the first $i$ primes as denominator are in the range of $\beta_{s'+i}$.
\end{proof}

Finally, we observe that w.r.t.~the generating sequence $(b_i)_{i<\omega}$ all pairs of unequal elements of $G_R$ can be recursively enumerated.

\begin{claim}
Equality in $(G_R,+)$ is co-r.e.
\end{claim}
\begin{proof}[Proof of the Claim.]
We run the algorithm generating $(b_i)_{i<\omega}$ and in step $s$ return all elements of the finite set $\Delta_{\beta_s}$.
As inequalities w.r.t $\beta_s$ yield inequalities w.r.t.~$(b_i)_{i<\omega}$, we only enumerate correct information.
Further, for every two elements $x,y$ of $G_R$ fix representations w.r.t. $(b_i)_{i<\omega}$ and $s'$ large enough such that not more than the first $s'$ of the $b_i$ occur in these representations, all of these have stabilized up to stage $s'$ and all coefficients in the representations take values between $-s'-1$ and $s'+1$. Then $x\neq y$ if and only if the tuple of their representations is in $\Delta_{\beta_{s'}}$.
\end{proof}
This finishes the proof of the theorem.
\end{proof}

\noindent
As there are $K$-recursive Martin-Löf random reals, we obtain the following corollary.

\begin{corollary}
There exists a co-r.e.~random subgroup of the rational numbers.
\end{corollary}

\begin{remark}
Proposition \ref{prop:recursivegeneratingsequence} implies, in particular, that if $R \leq_T K$ is Martin-L\"{o}f
random, then there cannot exist any generating sequence for $G_R$ with respect to which equality of members of $G_R$
is r.e.  Indeed, suppose that such a generating sequence $\beta$ did exist, so
that $E := \{(\sigma,\tau) \in \Z^{<\omega} \times \Z^{<\omega}\mid \sigma \cdot \beta\res_{|\sigma|-1} =
\tau \cdot \beta\res_{|\tau|-1}\}$ is r.e.
Fix any $\sigma_0 \in \Z^{<\omega}$ such that $\sigma_0 \cdot \beta_{|\sigma_0|-1} = 1$
(since $1 \in G_R$, such a $\sigma_0$ must exist). 
Then there is a strictly increasing recursive enumeration
$i_0,i_1,i_2,\ldots$ such that for all $j$, $i_j$ is the first $\ell$ found for which the following
hold:
(i) $\ell > i_{j'}$ whenever $j' < j$;
(ii) there are $n_{\ell} \geq 1$ and relatively prime positive integers $q,r$ with $p_{\ell} \nmid q$ and $p_{\ell} \nmid r$
such that for some $m$, $(q\sigma_0,I_m(rp_{\ell}^{n_{\ell}})) \in E$.
Note that
\begin{equation*}
\begin{split}
(q\sigma_0,I_m(rp_{\ell}^{n_{\ell}})) \in E &\Leftrightarrow q = (q\sigma_0) \cdot \beta_{|\sigma_0|-1} = I_m(rp_{\ell}^{n_{\ell}}) \cdot 
\beta_{m-1} = rp_{\ell}^{n_{\ell}}b_{m-1} \\
~& \Leftrightarrow b_{m-1} = qp_{\ell}^{-n_{\ell}}r^{-1}.
\end{split}
\end{equation*}    
The Martin-L\"{o}f randomness of $R$ implies that $\beta$ contains infinitely many terms of the form $\frac{q'}{r'p_{\ell'}^{n'_{\ell'}}}$ with $n'_{\ell'} \geq 1$, $q'$ and $r'$ relatively prime and positive, $p_{\ell'} \nmid q'$ and $p_{\ell'} \nmid r'$.  
Thus $i_j$ is defined for all $j$, and by Proposition \ref{prop:recursivegeneratingsequence} this contradicts
the Martin-L\"{o}f randomness of $R$.
\end{remark}

Further, a variation of the algorithm yields that equality of the proper rational part is r.e. on random groups.

\begin{theorem}\label{thm:reequalitymod1}
If $R \leq_T K$ is Martin-Löf random, 
then equality modulo 1 on $(G_R,+)$ is r.e. with respect to some generating sequence.
\end{theorem}
\def\proofreequalitymod1{
\begin{proof}
The construction of the generating sequence follows the construction of $(b_i)_{i<\omega}$ in the proof of Theorem~\ref{coRE} with the main difference that in step $s+1$ instead of making sure that in case of replacements no already enumerated inequalities are destroyed, we have to make sure that all equalities modulo $1$ that have been established in the first $s$ steps are preserved.
Formally, this reads as $E_{\beta_s}\subseteq E_{\beta_{s+1}}$ 
with
\begin{align*}
E_{(q_i)_{i<s+1}} = \{ \,(\sigma_0, \sigma_1) \in \,&(\,\{-s-1,\ldots,s+1\}^{s+1}\,)^2 \mid \\
& \sigma_0, \sigma_1 \text{ modulo $1$ represent the same element w.r.t. } (q_i)_{i<s+1} \,\}.
\end{align*}
As we have to preserve equality modulo 1 and each prime occurs at most once as basis of a denominator, we may use $0$ to replace the prime power fraction(s) if necessary. The rest of the proof works the same way.
\end{proof}
}
\proofreequalitymod1

\medskip
\noindent The next main result is concerned with the model-theoretic properties of random subgroups of
rationals.  We recall that two structures (in the model-theoretic sense) $M$ and $N$ with the same set $\sigma$ of non-logical symbols are \emph{elementarily equivalent} (denoted $M \equiv N$) iff they satisfy the same first-order sentences over 
$\sigma$; the \emph{theory} of a structure $M$ (denoted $\theory(M)$) is the set of all first-order sentences (over the set of non-logical symbols of $M$) that are satisfied by $M$.  The reader is referred to \cite{Mar02} for more background on model theory.  
We will prove a result that may appear a bit surprising: even though Martin-L\"{o}f random subgroups
of $(\mathbb{Q},+)$ (viewed as classes of integer sequence representations) are not computable,
any such subgroup is elementarily equivalent to $(\Z,+)$ - the additive group of integers - and thus
has a decidable theory. 
In other words, the incomputability of a random subgroup of rationals, at least according to the notion of ``randomness'' adopted in the present work, has little or no bearing on the decidability of its first-order properties.  We begin by showing that the theory of any subgroup $G$ of rationals reduces to that of the subgroup of
$(\mathbb{Q},+)$ generated by the set of all rationals either equal to $1$ or of the shape $p^{-n}$, where $p$ is a prime infinitely dividing $G$ and $n \in \N$.  Our proof of this fact rests on a sufficient criterion due to Szmielew \cite{Szm55} for the elementary 
equivalence of two groups; this result will be stated as it appears in \cite{Khi98}.

\begin{theorem}(\cite{Szm55}, as cited in \cite{Khi98})\label{thm:szmelemenequiv} 
Let $p$ be a prime number and $G$ be a group.  For all $n \geq 1$, $k \geq 1$
and elements $g_1,\ldots,g_k \in G$,
define $G[p^n] := \{ x \in G\mid p^n x = 0\}$ and the following predicate $C(p;g_1,\ldots,g_k)$:
\begin{equation*}
\begin{aligned}
~ & \mbox{$C(p;g_1,\ldots,g_k) \Leftrightarrow$ the images $g'_1,\ldots,g'_k$ of $g_1,\ldots,g_k$ in the factor group $\overline{G} := G/G[p^n]$ are} \\
~ & \mbox{such that $g'_1+p\overline{G},\ldots,g'_k+p\overline{G}$ are linearly independent in $\overline{G}/p\overline{G}$.}
\end{aligned}
\end{equation*}
Define the parameters $\alpha_{p,n}(G),\beta_p(G)$ and $\gamma_p(G)$ as follows.
\begin{equation*}
\begin{aligned}
\alpha_{p,n}(G) &:= \sup\{k \in \N\mid \mbox{$G$ contains $\Z_{p^n}^k$ as a pure subgroup} \}, \\
\beta_p(G) &:= \inf\{\sup\{k \in \N\mid \mbox{$\Z^k_{p^n}$ is a subgroup of $G$}\}\mid n \in \N \}, \\
\gamma_p(G) &:= \inf\{\sup\{k \in \N\mid (\exists x_1,\ldots,x_k)C(p;x_1,\ldots,x_k)\}\mid n \in \N\}.
\end{aligned}
\end{equation*}
(Here $pG := \{pg\mid g \in G\}$ and $\Z_{p^n}^k$ is the $k$-th power of the primary cyclic group on $p^n$ elements, 
that is, it consists of all elements $(a_0,\ldots,a_{k-1})$ such that $a_0,\ldots,a_{k-1} \in \Z_{p^n}$.)
Then any two groups $H$ and $L$ are elementarily equivalent iff $\alpha_{q,m}(H) = \alpha_{q,m}(L)$,
$\beta_{q}(H) = \beta_{q}(L)$ and $\gamma_{q}(H) = \gamma_{q}(L)$ for all primes $q$ and all $m \geq 1$.
\end{theorem}

\medskip
\noindent
The definition of a \emph{pure} subgroup will not be used in the proof of the subsequent theorem;
it will be observed that if $G$ is a subgroup of the rationals, then for $k \geq 1$ and $n \geq 1$, it cannot contain
$\Z_{p^n}^k$ as a subgroup in any case, 
so that 
$\alpha_{p,n}(G) = \beta_p(G) = 0$.


\begin{theorem}
Let $G$ be a subgroup of $(\mathbb{Q},+)$.
Then $G\equiv [\mathbb{Z}]_{P(G)}$, where $P(G) := \{i\in\N\mid(\forall x\in G)(\forall n\in\N) 
[\frac{x}{p_i^n}\in G]\}$ denotes the set of all primes infinitely dividing $G$ and for a set of 
primes $P$ we write $[\mathbb{Z}]_P$ for the subgroup of $(\mathbb{Q},+)$ generated by 
$\{1\}\cup\{\frac{1}{p^k} \mid p\in P, k\in \mathbb{N}\}$.
\end{theorem}
\begin{proof}
Define the predicate $C(p;x_1,\ldots,x_k)$ and the parameters $\alpha_{p,n},\beta_p$ and $\gamma_p$
as in Theorem \ref{thm:szmelemenequiv}. 
Let $p$ be a prime number and suppose $n \geq 1$.  By Theorem \ref{thm:szmelemenequiv}, it suffices 
to show that the three parameters $\alpha_{p,n},\beta_p$ and $\gamma_p$ coincide on $G$ and
$[\Z]_{P(G)}$. First, $\Z^k_{p^n}$ cannot be a subgroup of $G$ or $[\Z]_{P(G)}$ when $k \geq 1$
and $n \geq 1$ since by Theorem \ref{thm:subgroupsrationalschar}, no non-trivial subgroup of any
subgroup of rationals
can be torsion\footnote{We recall that a group $G$ is \emph{torsion} iff for every $x \in G$, there is some
$n$ such that $x^n$ is equal to the identity element of $G$.}; thus 
$\alpha_{p,n}$ and $\beta_p$ are both equal to $0$ for $G$ as well as $[\Z]_{P(G)}$. 
For a similar reason,
$G[p^n] := \{ x \in G\mid p^n x = 0\} = \{0\}$ for every $p$ and $n \in \N$,
and therefore $\overline{G} := G/G[p^n] = G/\{0\} = G$ and
$\overline{G}/p\overline{G} = G/pG$.  Furthermore, $G/pG$ may be regarded as a vector
space over the field $\Z_p$, and $(\exists x_1,\ldots,x_k)C(p;x_1,\ldots,x_k)$
holds iff the dimension of the $\Z_p$-vector space $G/pG$ (denoted
$\dim_{\Z_p}(G/pG)$) is at least $k$.  It
follows that
\begin{equation*}
\begin{aligned}
\gamma_p(G) &= \inf\{\sup\{k \in \N\mid (\exists x_1,\ldots,x_k)C(p;x_1,\ldots,x_k)\}\mid n\in\N\} \\
~ &= \inf\{\sup\{k \in \N\mid \dim_{\Z_p}(G/pG) \geq k\}\mid n\in\N\} \\
~ &= \dim_{\Z_p}(G/pG).
\end{aligned}
\end{equation*}
Similarly, $[\Z]_{P(G)}/p[\Z]_{P(G)}$ is a $\Z_p$-vector space and
$\gamma_p([\Z]_{P(G)}) = \dim_{\Z_p}\left([\Z]_{P(G)}/p[\Z]_{P(G)}\right)$.
Thus it suffices to show that $\dim_{\Z_p}(G/pG) = \dim_{\Z_p}\left([\Z]_{P(G)}/p[\Z]_{P(G)}\right)$.
\medskip

\noindent\textbf{Case 1:} $p \in P(G)$.  Then $pG = G$.  
It follows that $G/pG = G/G = \{0\}$; the same argument shows that $[\Z]_{P(G)}/p[\Z]_{P(G)} = \{0\}$.

\medskip
\noindent\textbf{Case 2:} $p \notin P(G)$.  Then there is some non-zero $x \in G$ such that $p^{-1} \cdot x \notin G$.  
It may be assumed without loss of generality that $x = 1$ because if
$x = u\cdot v^{-1}$ for some non-zero integers $u$ and $v$, then,
taking $G' = vG$, $G'$ is a subgroup of $(\mathbb{Q},+)$
that is isomorphic to $G$ such that $P(G) = P(G')$.
Assuming $p^{-1} \notin G$, there is a fixed integer $z$ such that $G$ is generated 
(as a subgroup of $(\mathbb{Q},+)$) by
rationals of the shape $zq^{-n}$, where $q$ ($\neq p$) is prime and $n \geq 1$.
As before, it may be assumed without loss of generality that $z = 1$.
It will be shown that each such generator is congruent to an integer modulo $pG$.
Fix a generator of the shape $q^{-n}$.  Let $m$ and $l$ be integers such that $mq^n + lp = 1$.
Then $q^{-n} = m + p\cdot\left( l \cdot q^{-n}\right)$.  It follows that $G/pG$ is isomorphic to $\Z_p$
and so $\dim_{\Z_p}(G/pG) = \dim(\Z_p) = 1$.  Using the case assumption that $p^{-1} \notin P(G)$, one
also has that $p^{-1} \notin [\Z]_{P(G)}$, and so the same argument as before yields
$\dim_{\Z_p}\left([\Z]_{P(G)}/p[\Z]_{P(G)}\right) = 1$.  
\end{proof}      

\noindent
Note that $\theory([\Z]_K,+)$ is undecidable; in contrast,
for $R$ Martin-L\"{o}f random we have
$P(G_R)=\varnothing$, so the promised corollary follows.

\begin{corollary}
Let $R \in 2^\omega$ be Martin-Löf random. 
Then $(G_R,+)$ and $(\mathbb{Z},+)$ have the same theories.
\end{corollary}

\noindent
One may ask whether this still holds for richer structures. This is not the case, as for example the theory of $(G,+,<)$ is different from $\theory(\mathbb{Z},+,<)$, as in the latter $x=1$ is a satisfying assignment for the formula $x+x>x \wedge \forall y<x\: \neg y+y>y$. There does not exist an $x\in G_R$ with this property for a ML-random $R$.



\section{Learning Finitely Generated Subgroups of a Random Subgroup of Rationals}
\label{sec:learnrandomsubgroups}

In this section, we investigate the learnability of non-trivial finitely generated subgroups of any 
group $G_R$ generated by a Martin-Löf random sequence $R$ such that $R \leq_T K$.
More specifically, we will examine for any given $G_R$ the set $F_{\beta}$ of 
representations of elements of any non-trivial finitely generated subgroup $F$ of $G_R$ with respect to a fixed generating 
sequence $\beta$ for $G_R$ such that all $F_{\beta}$ are r.e., and consider the learnability of the class 
of all such sets of representations.
        
We will consider learning from \emph{texts}, where a text is an infinite sequence that
contains all elements of $F_\beta$ for the $F$ to be learnt 
 and may contain the symbol $\#$, which indicates a pause in the data presentation and thus no new information.
For any text $T$ and $n \in \N$, $T(n)$ denotes the $(n+1)$-st term of $T$ and
$T[n]$ denotes the finite sequence $T(0),\ldots,T(n-1)$, i.e., the \emph{initial segment} of length 
$n$ of $T$; $\ran(T[n])$ denotes the set of non-pause elements occurring in $T[n]$.
A \emph{learner} $M$ is a recursive function mapping $(\Z^{<\omega}\cup\{\#\})^*$ into
$\N\cup\{?\}$; the $?$ symbol permits $M$ to abstain from conjecturing at any stage.  
A learner is fed successively with growing initial segments of the text and it
produces a sequence of conjectures $e_0,e_1,e_2,\ldots$, which are interpreted
with respect to a fixed \emph{hypothesis space}. 
In the present paper, we stick to the standard hypothesis space, a fixed G\"{o}del numbering $W_0,W_1,W_2,\ldots$ of all 
r.e.\ subsets of $\mathbb{Z}^{<\omega}$. In our setting from the generator $\frac{q}{m}$ of $F$ we can immediately derive an index $e$ for $F_\beta$ and therefore in the proofs we argue for learning $q$ and $m$.
The learner is said to \emph{behaviourally
correctly} (denoted $\Bc$) learn the representation $F_{\beta}$ of 
a finitely generated subgroup $F$ with respect to a fixed generating sequence $\beta$ for $G_R$ 
iff on every text for $F_{\beta}$, the sequence of conjectures output by the learner 
converges to a correct hypothesis; in other words, the learner almost always outputs an 
r.e.\ index for $F_{\beta}$ \cite{Fe72,CS83,Bar74}.  If almost all of the learner's 
hypotheses on the given text are equal in addition to being correct, then the learner is 
said to \emph{explanatorily} (denoted $\Ex$) learn $F_{\beta}$ (or it learns $F_{\beta}$ 
\emph{in the limit}) \cite{Go67}.

A useful notion that captures the idea of the learner converging on a given text 
is that of a \emph{locking sequence}, or more generally that of a 
\emph{stabilising sequence}.
A sequence $\sigma \in (\N \cup \{\#\})^*$ is called a \emph{stabilising sequence}
\cite{Fulk85} 
for a learner $M$ on some set $L$ if $\ran(\sigma) \subseteq L$ and for
all $\tau \in (L \cup \{\#\})^*$, $M(\sigma) = M(\sigma\circ\tau)$.
A sequence $\sigma \in (\N \cup \{\#\})^*$ is called a \emph{locking sequence}
\cite{BB75}
for a learner $M$ on some set $L$ if $\sigma$ is a stabilising sequence for
$M$ on $L$ and $W_{M(\sigma)} = L$. 

The following proposition due to Blum and Blum \cite{BB75} will be occasionally 
useful.

\begin{proposition}\label{prop:lockingsequence}\cite{BB75}
If a learner $M$ explanatorily learns some set $L$, then there exists a
locking sequence for $M$ on $L$.  Furthermore, all stabilising sequences for
$M$ on $L$ are also locking sequences for $M$ on $L$.
\end{proposition}

\noindent Clearly, also a $\Bc$-version of Proposition~\ref{prop:lockingsequence} holds.

It is not clear in the first place whether or not every finitely generated subgroup
of a randomly generated subgroup of $(\mathbb{Q},+)$ can even be represented
as an r.e.\ set.  This will be clarified in the next series of results.
We recall that a \emph{finitely generated subgroup} $F$ of $G_R$ is any subgroup
of $G_R$ that has some \emph{finite generating set} $S$, which means that every element of
$F$ can be written as a linear combination of finitely many elements of $S$ and
the inverses of elements of $S$.  $F$ is \emph{trivial} if it is
equal to $\{0\}$; otherwise it \emph{non-trivial}.  Furthermore, if $G_R$ is a subgroup of 
$(\mathbb{Q},+)$, then any finitely generated subgroup $F$ of $G_R$ is \emph{cyclic}, that is,
$F = \spn{\displaystyle\frac{q}{m}}$ for some $q \in \N$ and $m \in \N$ with $\gcd(q,m) = 1$
(see, for example, \cite[Theorem 8.1]{Szabosands}).  The latter fact will be
used freely throughout this paper.  For any generating sequence $\beta$ for $G_R$
and any finitely generated subgroup $F$ of $G_R$, the set of representations of elements of $F$
with respect to $\beta$ will be denoted by $F_{\beta}$. 

\begin{theorem}
\label{FinSubgroupsRE}
Let $R \leq_T K$ be Martin-Löf random. 
Then there is a generating sequence $(b_i)_{i<\omega}$ of $G_R$ such that for every non-trivial finitely generated subgroup $F$ of $G_R$ the set $F_\beta$ is r.e. 
\end{theorem}
\def\proofofFinSubgroupsRE{
\begin{proof}
We denote the set of all non-trivial finitely generated subgroups of $G_R$ by $\mathcal{F}$ and modify the construction of the generating sequence $(b_i)_{i<\omega}$ in the proof of Theorem~\ref{coRE}.
In contrast we show that for every $F \in \mathcal{F}$ there is some $s_F$ such that for every $s \geq s_F$ in step $s+1$ we can assure that replacements do not violate the property to represent an element of $F$,
i.e. it is possible to change entries of $\beta_s$ to obtain $\beta_{s+1}\res(s+1)$, such that we have $F_{\beta_s} \subseteq F_{\beta_{s+1}\res(s+1)}$, where
\begin{align*}
F_{(q_i)_{i<s+1}} = \{ \,\sigma \in \,&\{-s-1,\ldots,s+1\}^{s+1} \mid \\
& \sigma \text{ represents an element of $F$ w.r.t. } (q_i)_{i<s+1} \,\}.
\end{align*}
Let $F \in \mathcal{F}$, then there are $q$ and $m$ coprime, such that $F$ is generated by $\frac{q}{m}$.
Let $h \in \mathbb{N}$ be such that all prime factors of $q$ or $m$ are less or equal to $p_h$. We let $s_F \in \mathbb{N}$ be such that all $b_i$ having powers of a prime below $p_h$ as denominator have stabilized up to stage $s_F$ and the exponents occurring in the prime factorizations of $q$ and $m$ are $\leq s_F$.

We may assume that only the $j$-th component $\beta_s(j)=\frac{1}{p_\ell^{n_{\ell,s}}}$ for some $\ell>h$ has to be replaced by some integer $w$.
Thus, for all $\sigma \in F_{\beta_s}$ we want to make sure
$$\sum_{\overset{i=0,}{i\neq j}}^{s} \sigma(i) \beta_{s}(i) + \sigma(j)w \in F.$$

For this, it suffices to show that $\sigma(j)(w-\beta_s(j))\in F$.
By the Chinese Remainder Theorem there exists some integer $w$ such that for all $i<\ell$ we have $1 \equiv p_\ell^{n_{\ell,s}}w \mod p_i^{s}$.
With this there is some integer $z$ such that
$$\sigma(j)\left(w-\beta_s(j)\right)
=\frac{\sigma(j)}{{p_\ell^{n_{\ell,s}}}}\left(p_\ell^{n_{\ell,s}}w-1\right)
=z\:\frac{\sigma(j)}{{p_\ell^{n_{\ell,s}}}}\,\prod_{i<\ell}{p_i^{s}}.$$
Because $\ell > h$ we obtain that $\sigma(j)$ divided by $p_\ell^{n_{\ell,s}}$ is an integer and moreover $q$ is a factor of $\prod_{i<\ell}{p_i^{s}}$.
All integer-multiples of $q$ are members of $F$.
In a nutshell, enumerating
$\{ \sigma \in \mathbb{Z}^{<s_F} \mid \sigma\circ 0^{s_F-|\sigma|} \in F_{\beta_{s_F}} \}$
and all elements of $F_{\beta_s}$ for $s\geq s_F$ yields
the set of all representations of elements of $F$ w.r.t.~$(b_i)_{i<\omega}$.
\end{proof}
}
\proofofFinSubgroupsRE

\begin{remark}
The statement of Theorem \ref{FinSubgroupsRE} excludes the trivial subgroup because for any
generating sequence $\beta := (b_i)_{i < \omega}$ for $G_R$, $\spn{0}_{\beta}$ cannot be r.e.  To see this,
suppose, by way of contradiction, that $\spn{0}_{\beta}$ were r.e.  Given any $\sigma,\sigma'\in\Z^{<\omega}$,
set $\ell = \max(\{|\sigma|-1,|\sigma'|-1\})$, and for all $i \in \{0,\ldots,\ell\}$, $w_i = \sigma(i)$ if 
$i \leq |\sigma|-1$ and $0$ otherwise,
and $v_i = \sigma'(i)$ if $i \leq |\sigma'|-1$ and $0$ otherwise.
Then $\sigma \cdot \beta\res_{|\sigma|-1} = \sigma' \cdot \beta\res_{|\sigma'|-1} \Leftrightarrow \sigma \cdot \beta\res_{|\sigma|-1} - \sigma' \cdot \beta\res_{|\sigma'|-1} = 0 \Leftrightarrow \sum_{i=0}^{\ell} (w_i - v_i)b_i = 0  \Leftrightarrow (w_0-v_0,w_1-v_1,\ldots,w_{\ell}-v_{\ell}) \in \spn{0}_{\beta}$.
Thus if $\spn{0}_{\beta}$ were r.e., then equality with respect to $\beta$  would also be r.e., which, 
as was shown earlier, is impossible. 
\end{remark}

\noindent
We 
note that there cannot be any generating sequence $\beta$ for $G_R$
such that there are finitely generated subgroups $F,F'$ of $G_R$ with $F_{\beta}$ r.e.\ and $F'_{\beta}$
co-r.e.

\begin{theorem}\label{thm:nogenseqfinitereandcore}
Let $R \leq_T K$ be Martin-L\"{o}f random. 
Let $\beta$ be any generating sequence for $G_R$.  Then for any finitely generated subgroups
$F$ and $F'$ of $G_R$, one of the following holds: (i) both $F_{\beta}$ and $F'_{\beta}$ are r.e.,
(ii) both $F_{\beta}$ and $F'_{\beta}$ are co-r.e., or (iii) at least one of $F_{\beta}$ and $F'_{\beta}$ is
neither r.e.\ nor co-r.e.  
\end{theorem}

\def\proofofnogenseqfinitereandcore{
\begin{proof}
Fix any generating sequence $\beta := (b_i)_{i < \omega}$ for $G_R$.  Assume, by way of contradiction,
that for some $F = \spn{\displaystyle\frac{m'}{m}}$ and $F' = \spn{\displaystyle\frac{m''}{m'''}}$, where $m,m',m'',m''' \in \Z$
and $m,m'''>0$, $F_{\beta}$ is r.e.\ and $F'_{\beta}$ is co-r.e; without loss of generality, assume that $m''' = m$.  
It will be shown that this implies the existence of a strictly increasing recursive enumeration 
$i_0,i_1,i_2,\ldots$ such that $p_{i_j}^{-1} \in G_R$ for all $j$.
For all $j$, let $i_j$ be the first $\ell$ found such that $i_j > i_{j'}$ for all $j' < j$ and
there is some $\sigma \in \Z^{<\omega}$ such that the following conditions are satisfied.
\begin{enumerate}
\item $(m'p_{\ell}\sigma) \cdot \beta\res_{|\sigma|-1} \in F_{\beta}$.
\item For all $i \in \{1,\ldots,p_{\ell}-1\}$, $(m''i\sigma) \cdot \beta\res_{|\sigma|-1} \notin F'_{\beta}$.
\end{enumerate}
The Martin-L\"{o}f randomness of $R$ implies that there are arbitrarily large primes
$p$ with $p^{-1} \in G_R$.  For each $p^{-1} \in G_R$ such that $p > m$, there is some $\sigma_0 \in \Z^{<\omega}$
with $\sigma_0 \cdot \beta\res_{|\sigma_0|-1} = p^{-1}$, and so $m'p\sigma_0 \cdot \beta\res_{|\sigma_0|-1} = m' \in
\spn{F}$.  Moreover, for all $i \in \{1,\ldots,p-1\}$, since $p \nmid i$ and $p \nmid m$,
one has $m''i\sigma_0 \cdot \beta\res_{|\sigma_0|-1} = m''ip^{-1} \notin \spn{F'}$.
Hence $i_j$ is defined for all $j$.  Furthermore, suppose some prime $p_{\ell}$ and $\sigma_1 \in \Z^{<\omega}$
satisfy Conditions 1 and 2.  Condition 1 implies that $m'p_{\ell}\sigma_1 \cdot \beta\res_{|\sigma_1|-1}
\in \spn{\displaystyle\frac{m'}{m}}$, and so 
\begin{equation}\label{eqn:condition1}
p_{\ell}\sigma_1 \cdot \beta\res_{|\sigma_1|-1} \in \spn{m^{-1}}.
\end{equation}
Condition 2 implies that $m''\sigma_1 \cdot \beta\res_{|\sigma_1|-1} \notin \spn{\displaystyle\frac{m''}{m}}$, and
so 
\begin{equation}\label{eqn:condition2}
\sigma_1 \cdot \beta\res_{|\sigma_1|-1} \notin \spn{m^{-1}}.     
\end{equation}
It follows from (\ref{eqn:condition1}) and (\ref{eqn:condition2}) that if $\sigma_1 \cdot \beta\res_{|\sigma_1|-1}
= \displaystyle\frac{q}{r}$ for some relatively prime integers $q$ and $r$ with $r > 0$, then $q \neq 0$, 
$r \nmid m$ and $p_{\ell} \mid r$.
Consequently, $p_{\ell}^{-1} \in G_R$, as required.
But the existence of a strictly increasing recursive enumeration $i_0,i_1,i_2,\ldots$ such 
that $p_{i_j}^{-1} \in G_R$ for all $j$ contradicts Proposition \ref{prop:recursivegeneratingsequence}.
\end{proof}
}
\proofofnogenseqfinitereandcore

\begin{nota}\label{nota:finitegeneratedrepresent}
Let $R \leq_T K$ be Martin-Löf random 
and let $\beta := (b_i)_{i<\omega}$ be any generating sequence of $G_R$. 
For any subgroup $F$ of $G_R$, $F_{\beta}$ denotes the
set of all representations of elements of $F$ with respect to $\beta$, that is,
$F_{\beta} := \{\sigma \in \Z^{<\omega}\mid \sum_{i=0}^{|\sigma|-1}\sigma(i)b_i \in F\}$.
Furthermore, define $\cF_{\beta} := \{F_{\beta}\mid \mbox{$F$ is a non-trivial finitely generated subgroup of $G_R$}\}$.
\end{nota}

\def\proofoffiniteggbclearn{
It will be verified that $M$ is indeed a behaviourally correct learner for $\cF_{\beta}$. 
Let $F := \spn{\displaystyle\frac{q}{m}}$ be any finitely generated subgroup of
$G_R$, where $q$ and $m$ are relatively prime natural numbers, and let $T$ be any
text for $F_{\beta}$. 
Since every integer in $F$ is a multiple of 
$q$ and $T$ must contain $(q)$, it follows that after seeing a sufficiently long segment
$T[s_1]$ of $T$, $M$ will always correctly guess that the numerator of the target subgroup's 
generating element is equal to $q$.  

Suppose $m = r_0^{h_0}\ldots r_k^{h_k}$ for some 
positive integers $h_0,\ldots,h_k$ and primes $r_0,\ldots,r_k$ with $r_0 < \ldots < r_k$.  
For all $i \in \{0,\ldots,k\}$, $T$ contains an element of the shape $\left(0,\ldots,0,qr_i^{h''_i},
0,\ldots,\right.$ $\left.0\right)$, where, if the $(j+1)$-st coordinate of this element is the only non-zero entry, 
then $b_j = r_i^{-h'_i}$ for some $h'_i$ with $h'_i - h''_i = h_i$.  Consequently, for sufficiently 
large $s_2$, $r_0^{h_0}\ldots r_k^{h_k}$ divides $m_{s_2,t}$ whenever $t \geq s_2$ 
($m_{s_2,t}$ may be divisible by other prime powers as well).  Let $j_0,\ldots,j_k$ be such that for all 
$i \in \{0,\ldots,k\}$, $b_{j_i} = r_i^{-h'''_i}$ for some $h'''_i \geq 1$.  Fix $s_3 > \max(\{s_1,
s_2\})$ such that 
\begin{enumerate}[label = \roman*.]
\item for all $s \geq s_3$ and all $j \in \N$, if $b_{j,s} = p_i^{-e}$ for some $e \geq 1$
and prime $p_i \leq \max(\{q,r_k\})$, then $b_{j,s} = b_j$ (in other words, all entries of $\beta$
that are equal to $p_i^{-e}$ for some $e \geq 1$ and $p_i \leq \max(\{q,r_k\})$ have stabilised at
stage $s$); in particular, $b_{j_i,s} = b_{j_i}$ for all $i \in \{0,\ldots,k\}$;  
\item $s_3 > \max(\{i,e\})$ for all prime powers $p_i^e$ that are factors of either 
$q$ or $m$. 
\end{enumerate}

First, it will be shown that $W_{M(T[s])} \subseteq F_{\beta}$ if $s \geq s_3$.
Fix any $s \geq s_3$ and $t \geq s$.  By the choice of $s_3$, $r_0^{h_0}\ldots r_k^{h_k}$
divides $m_{s,t}$.  Suppose $m_{s,t} = r_0^{h_0}\ldots r_k^{h_k} g_0^{c_0} \ldots g_{\ell'}^{c_{\ell'}}$ 
for some positive integers $c_0,\ldots,c_{\ell'}$ and primes $g_0,\ldots,g_{\ell'}$ that do not divide $q$ or $m$.
Consider any $\sigma \in \Z^{<\omega}$ such that 
\begin{align}\label{eqn:sum1}
\sum_{i=0}^{|\sigma|-1} \sigma(i)b_{i,s'} \in \spn{qr_0^{-h_0}\cdots 
r_k^{-h_k}g_0^{-c_0}\cdots g_{\ell'}^{-c_{\ell'}}}
\end{align}
for some $s' \geq s$.  
It may be assumed without loss of generality that for 
any $p' \in \{r_0,\ldots,r_k,g_0,$ $\ldots,g_{\ell'}\}$, 
$\sum_{i=0}^{|\sigma|-1} \sigma(i)b_{i,s'} \notin 
\spn{qp' r_0^{-h_0}\ldots r_k^{-h_k}g_0^{-c_0}\ldots 
g_{\ell'}^{-c_{\ell'}}}$. 
Let $d_0,\ldots,d_{\ell'}$ be such that for all $i \in \{0,\ldots,\ell'\}$, $b_{d_i,s} =
b_{d_i,s'} = g_i^{-c'_i}$ for some $c'_i \geq c_i$; by the preceding assumption, 
$d_i \in \{0,\ldots,|\sigma|-1\}$ for all $i \in \{0,\ldots,\ell'\}$.     
We show that
\begin{align}\label{eqn:sum2}
\sum_{i=0}^{|\sigma|-1} \sigma(i)b_i \in \spn{qr_0^{-h_0}\cdots 
r_k^{-h_k}};
\end{align}
this will establish that $W_{M(T[s])} \subseteq F_{\beta}$.

The following relation will be established:
\begin{align}\label{eqn:sum3}
\sum_{i=0}^{|\sigma|-1} \sigma(i)b_i \in \spn{qr_0^{-h_0}\cdots r_k^{-h_k}g_0^{-c'_0}\ldots 
g_{\ell'}^{-c'_{\ell'}}}. 
\end{align}
It suffices to show that for every $i \in \{0,\ldots,|\sigma|-1\}$, $(b_i - b_{i,s'}) \sigma(i) \in 
\left\langle qr_0^{-h_0}\cdots r_k^{-h_k}g_0^{-c'_0}\cdots\right.$ $\left.g_{\ell'}^{-c'_{\ell'}}\right\rangle$; this 
fact, combined with (\ref{eqn:sum1}), will establish (\ref{eqn:sum3}). 

Pick any $i \in \{0,\ldots,|\sigma|-1\}$; without loss of generality, assume that  
$b_i \neq b_{i,s'}$. 
By the Martin-L\"{o}f randomness of $R$, it may be assumed that every $(i+1)$-st entry of $\beta$
(for any $i \in \N$) is either equal to $p_j^{-n'_j}$ for some $n'_j \geq 1$ or equal to some integer $p'$ such that
the $(i+1)$-st coordinate of $\beta$ is changed exactly once from some value $p_{j'}^{-n''_{j'}}$ (where 
$n''_{j'} \geq 1$) to $p'$; in addition, for any two terms of $\beta$ of the shape $p_{i_1}^{-n'_{i_1}}$
and $p_{i_2}^{-n'_{i_2}}$, where $n'_{i_1} \geq 1$ and $n'_{i_2} \geq 1$, $i_1 \neq i_2$.    
Thus $b_{i,s'} = p_j^{-n'_j}$ and $b_i = w_i$ for some $j \in \N, n'_j \geq 1$ and $w_i \in \Z$. 
\begin{description}[leftmargin=0cm]
\item[Case 1:] $p_j \notin \{r_0,\ldots,r_k,g_0,\ldots,g_{\ell'}\}$.
Then
\begin{equation*}
\begin{split}
\left(b_i - b_{i,s'}\right)\sigma(i) &= 
\left(w_i - p_j^{-n'_j}\right)\sigma(i) \\
&= \left(w_ip_j^{n'_j} - 1\right)\cdot \sigma(i)\cdot p_j^{-n'_j}. \\ 
\end{split}
\end{equation*}
By Conditions i and ii, as well as by the choice of $w_i$ (as given in the proof of Theorem \ref{FinSubgroupsRE}),
every prime power factor of $q$ must divide $w_ip_j^{n'_j}-1$; in particular, $q$ divides
$w_ip_j^{n'_j}-1$.  Furthermore,
by (\ref{eqn:sum1}) and the following two facts: (a) $p_j \notin \{r_0,\ldots,r_k,g_0,\ldots,g_{\ell'}\}$
and (b) $p_j$ does not divide $b^{-1}_{i',s'}$ for all $i' \neq i$ with $b^{-1}_{i',s'} \in \N$,
one has $\sigma(i)\cdot p_j^{-n'_j} \in \Z$.  
Therefore 
\begin{equation*}
\begin{split}
\left(b_i - b_{i,s'}\right)\sigma(i) &= 
\left(w_ip_j^{n'_j} - 1\right)\cdot \sigma(i)\cdot p_j^{-n'_j} \\
&\in \spn{q} \\
&\subseteq \spn{qr_0^{-h_0}\cdots r_k^{-h_k}g_0^{-c'_0}\cdots g_{\ell'}^{-c'_{\ell'}}}. 
\end{split}
\end{equation*}
\item[Case 2:] $p_j \in \{r_0,\ldots,r_k,g_0,\ldots,g_{\ell'}\}$.
By Condition i, $p_j \in \{g_0,\ldots,g_{\ell'}\}$; suppose $p_j = g_{i'}$ for some
$i' \in \{0,\ldots,\ell'\}$, so that $n'_j = c'_{i'}$.  As in Case 1,
\begin{equation*}
\begin{split}
\left(b_i - b_{i,s'}\right)\sigma(i) &= 
\left(w_i - g_{i'}^{-c'_{i'}}\right)\sigma(i) \\
&= \left(w_ig_{i'}^{c'_{i'}} - 1\right)\cdot \sigma(i)\cdot g_{i'}^{-c'_{i'}}. \\ 
\end{split}
\end{equation*}  
Conditions i and ii, together with the choice of $w_i$, imply that $q$ divides $w_i g_{i'}^{-c'_{i'}}-1$.
Thus, as before,
\begin{equation*}
\begin{split}
\left(b_i - b_{i,s'}\right)\sigma(i) &= 
\left(w_ig_{i'}^{c'_{i'}} - 1\right)\cdot \sigma(i)\cdot g_{i'}^{-c'_{i'}} \\
&\in \spn{qg_{i'}^{-c'_{i'}}} \\
&\subseteq \spn{qr_0^{-h_0}\cdots r_k^{-h_k}g_0^{-c'_0}\cdots g_{\ell'}^{-c'_{\ell'}}}. 
\end{split}
\end{equation*} 
\end{description}
This establishes (\ref{eqn:sum3}).
Now if $\sum_{i=0}^{|\sigma|-1} \sigma(i)b_i \in \spn{qr_0^{-h_0}\cdots r_k^{-h_k}g_0^{-c'_0}\ldots 
g_{\ell'}^{-c'_{\ell'}}} \sm  \spn{qr_0^{-h_0}\cdots r_k^{-h_k}}$, then
there must be a least $i'' \in \{0,\ldots,\ell'\}$ such that $b_{d_{i''},s'} = b_{d_{i''}} = g_{i''}^{-c'_{i''}}$.
But since $\ran(T)$ contains $(0,\ldots,0,qg_{i''}^{c''},0,\ldots,0)$,
where the $(d_{i''}+1)$-st position is the only non-zero entry and $c'_{i''} - c'' = c_{i''} \geq 1$,
$g_{i''}^{c_{i''}}$ must then be a factor of $m$, a contradiction.  Hence
$\sum_{i=0}^{|\sigma|-1} \sigma(i)b_i \in  \spn{qr_0^{-h_0}\cdots r_k^{-h_k}}$.     

Furthermore, since $m_{s,t''} = m$ for sufficiently large $t'' \geq s$ and for any given $l$,
there is some $t' \geq s$ with $b_{i,t'} = b_i$ whenever $i \leq l$, 
one also has that $F_{\beta} \subseteq W_{M(T[s])}$.
Thus $W_{M(T[s])} = F_{\beta}$, as required. 
}

\begin{theorem}\label{thm:finiteggbclearn}
Let $R \leq_T K$ be Martin-Löf random. 
Then there is a generating sequence $\beta$ of $G_R$ such that $F_{\beta}$ is
r.e.\ for every non-trivial finitely generated subgroup $F$ of $G_R$ and
$\cF_{\beta}$ is $\Bc$-learnable. 
\end{theorem}
\begin{proof} 
We will reuse the generating sequence $\beta := (b_i)_{i<\omega}$ for $G_R$ constructed in the proof of Theorem 
\ref{FinSubgroupsRE}.  For all $i,t \in \N$, let $b_{i,t}$ denote the $t$-th approximation to
the $(i+1)$-st element of $\beta$. 
Define a learner $M$ on any text $T$ as follows.  Let $s$ be the length of the text
segment seen so far.  First, let $a_0,\ldots,a_{\ell}$ be all the positive integers such that
for every $i \in \{0,\ldots,\ell\}$, there is some $\sigma \in \ran(T[s])$ for which 
$\sigma = (a_i)$.  If no such $a_i$ exists, then $M$ just outputs a default index, say an 
r.e.\ index for the set of representations for $\spn{1}$. 
Otherwise, $M$ uses 
$q' := \gcd(a_0,\ldots,a_{\ell})$ as its current guess for the numerator of the target 
subgroup's generating element. 
Next, define an approximation $m_{s,t}$ to the denominator of 
the target subgroup for every $t \geq s$ as follows.  Consider every element of $\ran(T[s])$ of 
the shape $(0,\ldots,0,q'p_i^{h_i},0,\ldots,0)$, where (1) $\gcd(q',p_i) = 1$,
(2) $q'p_i^{h_i}$ is the only non-zero coordinate of the element and it occurs in the $(j+1)$-st
position, (3) $b_{j,t} = p_i^{-h'_i}$ for some $h'_i \geq h_i$, and (4) $h_j$ is the smallest number 
$h''$ such that $\left(0,\ldots,0,q'p_i^{h''},0,\ldots,0\right) \in 
\ran(T[s])$ (as before, $q'p_i^{h''}$ is the only non-zero coordinate and it occurs in the $(j+1)$-st
position).  
Let $m_{s,t}$ be the product of all factors $p_i^{h'_j-h_j}$ such that $p_i,h'_j$
and $h_j$ satisfy items 1 to 4; if there is 
no such factor, then set $m_{s,t} = 1$.  $M$ outputs an index $e$ such that $W_e$ enumerates 
all $\sigma \in \Z^{< \omega}$ such that for some $t \geq s$ and $t' \geq s$, 
$\sum_{i=0}^{|\sigma|-1} \sigma(i)b_{i,t} \in \spn{\displaystyle\frac{q'}{m_{s,t'}}}$.      

\proofoffiniteggbclearn     
\end{proof}

\medskip
\noindent
The next result shows, in contrast to Theorem \ref{thm:finiteggbclearn}, that
if $R \leq_T K$ is Martin-Löf random, then, given \emph{any} generating sequence
$\beta$ for $G_R$ such that $F_{\beta}$ is r.e.\ for every non-trivial finitely generated
subgroup $F$ of $G_R$, the class $\cF_{\beta}$ is not explanatorily learnable.

\begin{theorem}\label{thm:finitegeneratednotex}
Let $R \leq_T K$ be Martin-Löf random. 
Suppose $\beta := (b_i)_{i<\omega}$ is a generating sequence for $G_R$ such
that for any non-trivial finitely generated subgroup $F$ of $G_R$, $F_{\beta}$ is r.e.
Then $\cF_{\beta}$ is not $\Ex$-learnable.
\end{theorem}
\def\proofoffinitegeneratednotex{
\begin{proof}
Assume, by way of contradiction, that such a learner $N$ did exist.
By Proposition \ref{prop:lockingsequence}, one could then find a locking
sequence $\gamma$ for $N$ on the set $\Z_{\beta}$ of representations of 
$\Z$ with respect to $\beta$.  We show that this implies the existence of a strictly increasing recursive enumeration 
$i_0,i_1,i_2,\ldots$ such that for all $j$, $p_{i_j}^{-1} \in G_R$.
The enumeration $i_0,i_1,i_2,\ldots$ is defined as follows.  For each $j$, let $i_j$ be the first $\ell'$  found such that
$\ell' > i_{j'}$ for all $j' < j$ and there is a sequence $\delta \in (\Z^{<\omega})^*$ satisfying the following conditions.
\begin{enumerate}
\item $N(\gamma\circ\delta) \neq N(\gamma)$. 
\item For all $\sigma \in \ran(\delta)$, $p_{\ell'}\sigma \in \Z_{\beta}$.
\end{enumerate}  
Note that Condition 2 is semi-decidable because $\Z_{\beta}$ is r.e.
By the Martin-L\"{o}f randomness of $R$, there exist infinitely many $p$ such that
$p^{-1} \in G_R$.  For each such $p$, since $N$ must explanatorily learn $\spn{p^{-1}}_{\beta}$,
$W_{N(\gamma)} = \Z_{\beta} \subset \spn{p^{-1}}_{\beta}$ and 
$\ran(\gamma) \subseteq \Z_{\beta} \subset \spn{p^{-1}}_{\beta}$,
there exists some $\delta \in \left(\spn{p^{-1}}_{\beta}\right)^*$ such that $N(\gamma\circ\delta) \neq N(\gamma)$.
Furthermore, for each $\sigma \in \ran(\delta)$, one has 
\begin{equation*}
\begin{aligned}
(p\sigma) \cdot \beta_{|\sigma|-1} &= p (\sigma \cdot \beta_{|\sigma|-1}) \\
&\in p \spn{p^{-1}} \\
&\subseteq \Z.
\end{aligned}
\end{equation*}
Thus $i_j$ is defined for all $j$.  It remains to show that for all $j$,
$p_{i_j}^{-1} \in G_R$.  To see this, one first observes that by the locking 
sequence property of $\gamma$, if $\delta$ is the sequence found together with
$i_j$ satisfying Conditions 1 and 2, then $N(\gamma \circ \delta) \neq N(\gamma)$ implies
that there exists some $\sigma \in \ran(\delta)$ with $\sigma \notin \Z_{\beta}$;
in other words, $\sigma \cdot \beta_{|\sigma|-1} \notin \Z$.
By Condition 2, $p_{i_j}\sigma \in \Z_{\beta}$ and therefore $\sigma \cdot \beta_{|\sigma|-1}$
must be of the shape $qp_{i_j}^{-1} \in G_R$ for some $q \in \Z$ that is coprime to $p_{i_j}$.
Consequently, $p_{i_j}^{-1} \in G_R$, as required. 
But by Proposition \ref{prop:recursivegeneratingsequence}, the existence of 
the enumeration $i_0,i_1,i_2,\ldots$ would contradict the fact that $R$ is Martin-L\"{o}f random.  
Hence $\cF_{\beta}$ cannot be explanatorily learnable. 
\end{proof}
}
\proofoffinitegeneratednotex

\noindent
The next theorem considers the learnability of the set of representations of any finitely generated subgroup
$F$ of the quotient group $G_R/\Z$ with respect to the generating sequence for
$G_R/\Z$ constructed in the proof of Theorem \ref{thm:reequalitymod1}.  Slightly abusing the notation defined
in Notation \ref{nota:finitegeneratedrepresent}, for any
generating sequence $\beta$ for $G_R/\Z$, $F_{\beta}$ will denote the set of representations of any subgroup 
$F$ of $G_R/\Z$ with respect to $\beta$, and $\cF_{\beta}$ will denote $\{F_{\beta}\mid 
\mbox{$F$ is a finitely generated subgroup of $G_R/\Z$}\}$. 

\begin{theorem}\label{thm:finitelygeneratedmod1bc}
Suppose $R \leq_T K$ is Martin-Löf random. 
Let $G_R/\Z$ be the quotient group of $G_R$ by $\Z$.  Then there is a generating
sequence $\beta$ for $G_R/\Z$ such that $F_{\beta}$ is r.e.\ for all finitely
generated subgroups of $G_R/\Z$ and $\cF_{\beta}$ 
is $\Bc$-learnable. 
\end{theorem}
\def\proofoffinitelygeneratedmod1bc{
\begin{proof} 
We will use the fact that any finitely generated subgroup of $G_R/\Z$ is finite\footnote{This may be seen as follows.  Suppose $F = \spn{\displaystyle\frac{p}{q}}$ for some relatively prime $p,q \in \N$.  Then for all $n \in \N$, there are $n' \in \Z$ and $r$ with $0 \leq r < q$ such that $\displaystyle\frac{np}{q} = \displaystyle\frac{n'q+r}{q} \equiv \displaystyle\frac{r}{q}~(mod~1)$.}  (see, for example, \cite[page 106]{SteinSha02}).  Let $\beta := (b_i)_{i < \omega}$ be the generating sequence for $G_R/\Z$ constructed in the proof of Theorem \ref{thm:reequalitymod1}; as was shown in the proof of this theorem, equality is r.e.\ with respect to $\beta$, that is, $E := \{(\sigma,\sigma') \in \Z^{<\omega} \times \Z^{<\omega}\mid \sum_{i=0}^{|\sigma|-1} \sigma(i)b_i - \sum_{j=0}^{|\sigma'|-1} \sigma'(j)b_j \equiv 0~(mod~1)\}$ is r.e.  Then for any finitely generated subgroup $F$ of $G_R/\Z$
with elements $x_0,\ldots,x_k$, if $\sigma_i$ is a representation for $x_i$ for all $i \in \{0,\ldots,k\}$,
then $F_{\beta} = \bigcup_{0\leq i \leq k}\{ \tau \in \Z^{<\omega}\mid (\sigma_i,\tau) \in E\}$ is r.e.  Define a learner $M$ on any text $T$ as follows.  On input $T[s]$, $M$ outputs an r.e.\ index for the closure under equality of all $\sigma \in \ran(T[s])$, that is, $W_{M(T[s])} = \{\tau \in \Z^{<\omega}\mid (\exists \sigma \in \ran(T[s]))[(\sigma,\tau) \in E]\}$.  Let $F$ be any finitely generated subgroup of $G_R/\Z$, and suppose
$M$ is fed with a text $T$ for $F_{\beta}$.  By construction, $M$ always conjectures a set that is contained in $F_{\beta}$.
Furthermore, since $F$ is finite, there is a sufficiently large $s$ such that for all $x \in F$, $\ran(T[s])$ contains 
some $\sigma$ with $\sum_{i=0}^{|\sigma|-1} \sigma(i)b_i \equiv x~(mod~1)$.  Thus, as $M$ always conjectures a set
that is closed under equality with respect to $\beta$, it follows that for all $s' \geq s$, $M$ on $T[s']$ will conjecture $F_{\beta}$.
\end{proof}
}
\proofoffinitelygeneratedmod1bc

\noindent
As in the case of the collection of non-trivial finitely generated subgroups of $G_R$, the class $\cF_{\beta}$ is not explanatorily 
learnable with respect to any generating sequence $\beta$ for $G_R/\Z$.  The proof is entirely analogous to that of Theorem
\ref{thm:finitegeneratednotex}.

\begin{theorem}\label{thm:finitegeneratedmod1notex}
Let $R \leq_T K$ be Martin-Löf random. 
Suppose $\beta := (b_i)_{i<\omega}$ is a generating sequence for $G_R/\Z$ such
that for any finitely generated subgroup $F$ of $G_R/\Z$, $F_{\beta}$ is r.e.
Then $\cF_{\beta}$ is not $\Ex$-learnable.
\end{theorem}

\noindent
A natural question is whether the learnability or non-learnability of a class of representations for a collection of 
subgroups of $G_R$ is independent of the choice of the generating sequence for $G_R$.  We have seen in
Theorem \ref{thm:finitegeneratednotex}, for example, that the non explanatory learnability of the class of non-trivial finitely generated subgroups of $G_R$ holds for \emph{any} generating sequence for $G_R$ such that $F_{\beta}$ is r.e.\ whenever
$F$ is a finitely generated subgroup.    
The next theorem gives a positive learnability result that is to some extent independent of the choice of the generating 
sequence: for any generating sequence $\beta$ for $G_R$ such that equality with respect to $\beta$ is $K$-recursive and $F_{\beta}$ is r.e.\ whenever $F$ is a finitely generated subgroup of $G_R$, the class $\cF_{\beta}$ is explanatorily 
learnable relative to oracle $K$. 

\begin{theorem}\label{thm:exklearnfinite}
Let $R \leq_T K$ be Martin-L\"{o}f random. 
Then for any generating sequence $\beta$ for $G_R$ such that equality with respect to $\beta$ is $K$-recursive 
(in other words, the set $E_{\beta}:=\{(\sigma,\sigma') \in \Z^{<\omega} \times\Z^{<\omega}\mid \sigma \cdot 
\beta_{|\sigma|-1} = \sigma' \cdot \beta_{|\sigma'|-1}\}$ is $K$-recursive) and $F_{\beta}$ is r.e.\ for all finitely 
generated subgroups of $G_R$, $\cF_{\beta}$ 
is $\Ex[K]$-learnable.
\end{theorem}
\def\proofofexklearnfinite{
\begin{proof}
Let $\beta$ be any generating sequence for $G_R$ satisfying the hypothesis of the theorem.
Define an $\Ex[K]$ learner $M$ as follows.  On input $a_0 \circ \ldots \circ a_n$, where $a_i \in \Z^{<\omega} \cup 
\{\#\}$ for all $i \in \{0,\ldots,n\}$, oracle $K$ is first used to determine a representation $\rho$ of a generator for the 
subgroup generated by $\{a_i \cdot \beta_{|a_i|-1}\mid 0 \leq i \leq n \wedge a_i \notin \{\ve,\#\}\}$.  
This can be done in a recursive fashion.  We first identify the indices $i_0,\ldots,i_{\ell} \in \{0,\ldots,n\}$ (if any) such that
$a_i \notin \{\ve,\#\}$ and $(a_i,\textbf{0}) \notin E_{\beta}$ (here $\textbf{0}$ denotes any zero vector); if no such index 
exists, then let $\rho$ be any representation of $0$.  Set $\rho_0 = a_{i_0}$.  Having defined $\rho_k$, use oracle $K$ to determine relatively prime integers $i,j$ such that $(i \rho_k, j a_{i_{k+1}}) \in E_{\beta}$; without loss of generality,
assume that $j \geq 1$.  Then search for some $\rho' \in \Z^{<\omega}$ with $(j\rho',\rho_k) \in E_{\beta}$, and set 
$\rho_{k+1} = \rho'$.  Assuming inductively that $\rho_k$ represents a generator for the subgroup generated
by $\{a_{i_p}\cdot \beta_{|a_{i_p}|-1}\mid 0 \leq p \leq k\}$, one deduces from the relations $(i \rho_k, j a_{i_{k+1}}) 
\in E_{\beta}$ and $(j\rho',\rho_k) \in E_{\beta}$ that
\begin{equation*}
\begin{aligned}
a_{i_{k+1}} \cdot \beta_{|a_{i_{k+1}}|-1} &= \displaystyle\frac{i}{j} \rho_k \cdot \beta_{||\rho_k||-1} \\
&= \displaystyle\frac{i}{j} j \rho' \cdot \beta_{|\rho'|-1} \\
&= i \rho' \cdot \beta_{|\rho'|-1}.  
\end{aligned}
\end{equation*} 
Hence $a_{i_{k+1}} \cdot \beta_{|a_{i_{k+1}}|-1} \in \spn{\rho' \cdot \beta_{|\rho'|-1}}$.  Since
$\rho_k \cdot \beta_{|\rho_k|-1} \in \spn{\rho' \cdot \beta_{|\rho'|-1}}$, it follows from the
induction hypothesis that for all $l \leq k$, $a_{i_l} \cdot \beta_{|a_{i_l}|-1} \in \spn{\rho' \cdot \beta_{|\rho'|-1}}$.
Thus, setting $\rho = \rho_{k+1}$, $\rho \cdot \beta_{|\rho|-1}$ generates the subgroup generated
by $\{a_i \cdot \beta_{|a_i|-1}\mid 0 \leq i \leq n \wedge a_i \notin \{\ve,\#\}\}$.       

$M$ now outputs the least $e \leq n$ (if any such $e$ exists) such that the following hold:
\begin{enumerate}
\item $\ran(a_0\ldots a_n) \subseteq W_e$.
\item For all $\tau \in W_{e,n}$, there is some integer $q$ such that $(\tau,q\rho) \in E_{\beta}$.   
\end{enumerate}           
If there is no $e \leq n$ satisfying all of the above conditions, then $M$ outputs a default index, say $0$.
Suppose $M$ is fed with a text for the set of representations of some finitely generated subgroup $F$. 
Then $M$ will identify a generator $g$ such that $F = \spn{g}$ in the limit; Condition 1 ensures that in the
limit, $M$ will conjecture a set $W_e$ such that $\spn{g}_{\beta} \subseteq W_e$, while Condition 2 ensures
that in the limit, $M$ will not overgeneralise, that is, it will not output a set containing elements not in $\spn{g}_{\beta}$.
Hence $M$ explanatorily learns $F_{\beta}$ relative to oracle $K$.    
\end{proof}
}
\proofofexklearnfinite

\noindent
We recall from Theorem \ref{thm:reequalitymod1} that there is a generating sequence
$\beta := (b_i)_{i<\omega}$ for $G_R$ such that equality modulo $1$ with respect to $\beta$ is r.e.; in other words,
the set $\{(\sigma,\sigma') \in \Z^{<\omega} \times\Z^{< \omega}\mid \sigma \cdot \beta_{|\sigma|-1} \equiv \sigma' \cdot 
\beta_{|\sigma|'-1}~(mod~1) \}$ is r.e.  The next result considers the learnability of a class that is in some sense
``orthogonal'' to the class $\Z_{\beta}$: 
the class of all sets of representations of $\Z$ with respect to \emph{any} generating sequence $\beta'$ for $G_R$
such that $\Z_{\beta'}$ is r.e.  Equivalently, we ask whether the collection of all r.e.\ sets of pairs $(\sigma,\sigma')$ 
for which equality modulo $1$ holds with respect to any given generating sequence for $G_R$ can be learnt; it
turns out that this class is not even behaviourally correctly learnable.  In the statement
and proof of the next theorem, for any generating sequence
$\beta$ for $G_R$, let $E_{\beta}$ denote the set $\{(\sigma,\sigma') \in \Z^{<\omega} \times \Z^{<\omega}\mid
\sigma \cdot \beta_{|\sigma|-1} = \sigma' \cdot \beta_{|\sigma'|-1}~(mod~1)\}$.

\begin{theorem}\label{thm:reequalitynotbclearn}
Let $R \leq_T K$ be Martin-L\"{o}f random. 
Let $\cG_0$ be the collection of all generating sequences $\beta$ for $G_R$ such that 
$E_{\beta}$ is r.e., and define $\cE_0 := \{E_{\beta}\mid \beta \in \cG_0\}$.
Then $\cE_0$ is not $\Bc$-learnable.
\end{theorem}
\def\proofofreequalitynotbclearn{
\begin{proof}
Assume, by way of contradiction, that $\cE_0$ has a behaviourally correct learner $N$. 
Using a standard type of argument in inductive inference, we will build a limiting recursive generating sequence $\beta$ for 
$G_R$ and a text $T$ for $E_{\beta} := \{(\sigma,\sigma') \in \Z^{<\omega} \times \Z^{< \omega}\mid \sigma \cdot 
\beta_{|\sigma|-1} \equiv \sigma' \cdot \beta_{|\sigma'|-1}~(mod~1)\}$ such that $E_{\beta}$ is r.e.\ and $N$ on $T$ outputs 
some wrong conjecture for $E_{\beta}$ infinitely often, that is, there are infinitely many $s$ for which 
$W_{N(T[s])} \neq E_{\beta}$.

The basic construction of $\beta$ follows the proof of Theorem \ref{thm:reequalitymod1}, the main difference
being that at various stages of the construction, one searches for some sequence $\Gamma$ to extend
the current text segment $T_s$ such that $N$ on $T_s \circ \Gamma$ conjectures some r.e.\ set containing
a pair $(\sigma,\sigma') \in \Z^{<\omega} \times \Z^{<\omega}$ such that $\sigma$ and $\sigma'$ are both of
the shape $(0,\ldots,0,1)$, where $|\sigma| \neq |\sigma'|$ and $\sigma,\sigma'$ are both longer than any 
$\tau \in \ran(T_s\circ\Gamma)$; 
one may then ensure that $W_{N(T_s \circ \Gamma)}$ is a wrong conjecture by setting the $|\sigma'|$-th position of $\beta$ to 
$0$ and the $|\sigma|$-th position of $\beta$ to $p^{-1}$ for some fixed prime $p$ with $p^{-1} \in G_R$.  The constructions 
of $\beta$ and $T$ are given in more detail below. 
The approximations of $\beta$ and $T$ at stage $s$ will be denoted by $\beta^s$ and $T_s$ respectively. 
Fix some prime $p$ such that $p^{-1} \in G_R$.
              
\begin{enumerate}
\item Set $T_0 = \ve$ and $\beta^0 = \ve$.
\item At stage $s+1$, 
let $\gamma$ be a generating sequence for $G_R$ that extends a prefix of $\beta^s$ and is defined as in
Theorem \ref{thm:reequalitymod1}, so that equality modulo $1$ with respect to $\gamma$ is r.e.  
In other words, one builds $\gamma$ in increasing segments by searching at every stage $t$ a new element $p_j^{-n_{j,t}}
\in G_{R^t}$ such that $n_{j,t} \geq 1$, where $R^t$ is the $t$-th approximation to $R$,  and adding $p_j^{-n_{j,t}}$ as a 
new term to the current approximation of $\gamma$.  
Furthermore, for every $\ell$ such that the $\ell$-th term of the current approximation of $\gamma$ does not belong to 
$G_{R^t}$, the $\ell$-th term of $\gamma$ is permanently set to $0$.  
Note that an r.e.\ index for $E_{\gamma}$ can be uniformly computed from $\beta^s$.
Now search for some $\delta \in (E_{\gamma} \cup \{\#\})^*$ such that $W_{N(T_s \circ \delta)}$ enumerates a
pair $(\sigma,\tau)$ satisfying the following:
\begin{enumerate}
\item $|\sigma| \neq |\tau|$.
\item $|\sigma| > |\beta^s|$ and $|\tau| > |\beta^s|$.
\item For all $(\eta,\eta') \in \ran(\delta)$, $|\sigma| > \max(\{|\eta|,|\eta'|\})$ and $|\tau| > \max(\{|\eta|,|\eta'|\})$. 
\item The first $|\sigma|-1$ (resp.~$|\tau|-1$) terms of $\sigma$ (resp.~$\tau$) are equal to $0$ while
the last term of $\sigma$ (resp.~$\tau$) is equal to $1$. 
\end{enumerate}    
Since $E_{\gamma}$ is r.e.\ by construction, $N$ must behaviourally correctly learn $E_{\gamma}$.
Moreover, since every term of $\gamma$ is either an element of $G_R$ of the shape $p_i^{-n_i}$ for some $n_i \geq 1$ 
or equal to $0$,  Proposition \ref{prop:recursivegeneratingsequence} implies that $\gamma$ has infinitely many terms equal to 
$0$.  Hence such $\delta$ and $(\sigma,\tau)$ must eventually be found. 

\medskip
Without loss of generality, assume that $|\sigma| < |\tau|$.
Now let $\beta^{s+1}$ be the sequence of length $|\tau|+1$ such that the $(|\tau|+1)$-st position of $\beta^{s+1}$
is equal to $p_i^{-n_{i,t}}$ for the least $t \geq s$ such that for some minimum $i$, $n_{i,t} \geq 1$ and $\beta^s$ does not contain $p_i^{-n_{i,t}}$, the $|\tau|$-th position of $\beta^{s+1}$ is $p^{-1}$ and the terms of $\beta^{s+1}$ between the 
$|\sigma|$-th and $(|\tau|-1)$-st positions inclusive are all equal to $0$, and the terms of $\beta^{s+1}$ between
the $(|\beta^s|+1)$-st and the $(|\sigma|-1)$-st positions inclusive are equal to the respective terms of the $(s+1)$-st 
approximation of $\gamma$; furthermore, all the terms of $\beta^{s+1}$ are corrected up the $(s+1)$-st approximation,
that is, every term of $\beta^{s+1}$ belongs to $G_{R^{s+1}}$ and is either equal to $0$ or is of the shape
$p_i^{-n_i}$ for some $n_i \geq 1$.  Let $\Gamma$ be a string whose range consists of all $(\sigma,\sigma') \in 
\{-s-1,-s,\ldots,s,s+1\}^{<s+2} \times \{-s-1,-s,\ldots,s,s+1\}^{<s+2}$ such that $\sigma \cdot 
\beta^{s+1}_{|\sigma|-1} = \sigma' \cdot \beta^{s+1}_{|\sigma'|-1}$, and set $T_{s+1} = T_s \circ \delta \circ \Gamma$.   
\end{enumerate}
Set $T = \lim_{s\rightarrow\infty}T_s$ and $\beta = \lim_{s\rightarrow\infty}\beta^s$
(more precisely, for each $i \in \N$, $T(i) = \lim_{s\rightarrow\infty}T_s(i)$ and $\beta(i) = \lim_{s\rightarrow\infty}\beta^s(i)$).
Arguing as in the proof of Theorem \ref{thm:reequalitymod1}, the set $E_{\beta}$ is r.e.;
in addition, the range of $T$ is precisely equal to $E_{\beta}$.  On the other hand, by construction $N$ on
$T$ infinitely often outputs an r.e.\ index for some set not equal to $E_{\beta}$.
Hence $N$ cannot be a behaviourally correct learner for $\cE_0$. 
\end{proof}
}
\proofofreequalitynotbclearn

\noindent
In contrast to Theorem \ref{thm:reequalitynotbclearn}, we present a positive learnability result for the
collection of all co-r.e.\ sets of pairs of representations of $G_R$ for which equality holds with respect to any 
generating sequence for $G_R$.  In the statement and proof of the next theorem, given any generating sequence
$\beta$ for $G_R$ such that equality with respect to $\beta$ is co-r.e., $E_{\beta}$ will denote the 
set $\{(\sigma,\sigma') \in \Z^{<\omega} \times \Z^{<\omega}\mid\sigma \cdot \beta_{|\sigma|-1} = \sigma' \cdot 
\beta_{|\sigma'|-1}\}$.

\begin{theorem}\label{thm:e1exklearn}
Let $R \leq_T K$ be Martin-L\"{o}f random. 
Let $\cG_1$ be the collection of all generating sequences $\beta$ for $G_R$ such that 
$E_{\beta}$ is co-r.e., and define $\cE_1 := \{E_{\beta}\mid \beta \in \cG_1\}$.
Then $\cE_1$ is explanatorily learnable relative to oracle $K$ using co-r.e.\ indices.  That is to say, there is a 
$K$-recursive learner $M$ such that for any $E_{\beta} \in \cE_1$ and any text $T$ for $E_{\beta}$, 
$M$ on $T$ will output an r.e.\ index for $(\Z^{<\omega} \times\Z^{<\omega}) \sm E_{\beta}$ in the limit.  
\end{theorem}
\def\proofofe1exklearn{
\begin{proof}
Define a $K$-recursive learner $M$ for $\cE_1$ as follows.
On input $\gamma$, $M$ first guesses the minimum $e$ such that the $(e+1)$-st term of the generating
sequence is non-zero; it takes $e$ to be the smallest $e' \leq |\gamma|$ such that $(\textbf{0}, I_{e'+1}(1))
\notin \ran(\gamma)$ (where $\textbf{0}$ denotes the zero vector of length $1$; $I_{e+1}(1)$ denotes the vector of 
length $e+1$ whose first $e$ coordinates are $0$ and whose last coordinate is $1$); if no such $e'$ exists,
then $M$ outputs a default index, say an r.e.\ index for $\emptyset$.  Based on the
$(e+1)$-st term $b_e$ of the generating sequence and $\ran(\gamma)$, $M$ calculates some of the remaining 
terms of this sequence as a multiple of $b_e$.  For each $(\sigma,\tau) \in \ran(\gamma)$ such that
there are $d \in \N$, $q \in \Z$ and $r \in \Z^+$ with $\sigma = I_{e+1}(q)$, $\tau = I_{d+1}(r)$,
$d \neq e$ and $\gcd(q,r) = 1$, the $(d+1)$-st term $b_d$ of the generating sequence is $\displaystyle\frac{qb_e}{r}$.
Let $e_0,\ldots,e_{\ell}$ be all the numbers such that for each $i \in \{0,\ldots,\ell\}$,
$M$ has determined a rational number $q_i$ for which the $(e_i+1)$-st term of the generating sequence
equals $q_i b_e$ (in particular, there is a $j$ with $e_j = e$).  $M$ finds all pairs
$(\sigma_0,\sigma'_0),\ldots,(\sigma_{\ell'},\sigma'_{\ell'}) \in \Z^{<\omega} \times \Z^{<\omega}$ for 
which all non-zero positions of $\sigma_i$ and $\sigma'_i$ belong to $\{e_0,\ldots,e_{\ell}\}$ and 
$\sum_{j=0}^{\ell}\sigma_i(e_j)\cdot q_j \neq \sum_{j=0}^{\ell}\sigma'_i(e_j)\cdot q_j$. 
$M$ outputs the least index $c \leq |\gamma|$ satisfying the following conditions (if such a $c$ exists).
\begin{enumerate}
\item $\ran(\gamma) \cap W_c = \emptyset$.
\item For all $i \in \{0,\ldots,\ell'\}$, $(\sigma_i,\sigma'_i) \in W_c$.
\end{enumerate}           
If no such $c$ exists, then $M$ conjectures $\emptyset$.

Suppose $M$ is presented with a text $T$ for some $E_{\beta} \in \cG_1$, where $\beta$ is a generating
sequence for $G_R$ such that equality is co-r.e.\ with respect to $\beta$.  Suppose $\beta = (b_i)_{i < \omega}$.
Then $M$ on $T$ will find in the limit the least number $e$ such that $b_e \neq 0$
(since for all $d < e$, $(\textbf{0},I_{d+1}(1)) \in \ran(T)$).  By Condition 1, $M$ will, in the limit, always conjecture a set 
contained in $(\Z^{<\omega} \times \Z^{<\omega}) \sm E_{\beta}$.  
Furthermore, for all $(\sigma,\sigma') \notin E_{\beta}$ and every $i \leq \max(\{|\sigma|,|\sigma'|\})$, there are integers 
$q_i,r_i$ with $r_i > 0$ and $\gcd(q_i,r_i) = 1$ such that $q_ib_e = r_ib_i$, and therefore
$\ran(T)$ must contain $(I_{e+1}(q_i),I_{i+1}(r_i))$.  Since $b_e \neq 0$, one has 
$$
\sum_{j=0}^{|\sigma|-1}\sigma(j) \cdot \displaystyle\frac{q_j}{r_j} \neq \sum_{j=0}^{|\sigma'|-1}\sigma'(j) \cdot \displaystyle\frac{q_j}{r_j}
\Leftrightarrow
\sum_{j=0}^{|\sigma|-1}\sigma(j) \cdot \displaystyle\frac{q_jb_e}{r_j} \neq \sum_{j=0}^{|\sigma'|-1}\sigma'(j) \cdot \displaystyle\frac{q_jb_e}{r_j} 
\Leftrightarrow
(\sigma,\sigma') \notin E_{\beta},
$$
and thus by Condition 2, $M$ will, in the limit, always conjecture a set containing $(\sigma,\sigma')$. 
$M$ will therefore converge to the least index $c$ satisfying $E_{\beta} = (\Z^{<\omega} \times \Z^{<\omega}) \sm
W_c$, as required. 
\end{proof}
}
\proofofe1exklearn

\section{Random Subrings of Rationals and Random Joins of Pr\"{u}fer Groups}

We have seen in Section \ref{sec:StringRandomGroups} that any Martin-L\"{o}f random sequence $R \leq_T K$ 
gives rise to a random subgroup $G_R$ of rationals such that for some generating sequence $\beta$ for $G_R$,
equality with respect to $\beta$ is co-r.e.\ and another $\beta$ such that the set of representations of any non-trivial 
finitely generated subgroup of $G_R$ with respect to $\beta$ is r.e.  The present section will define other random 
structures with similar properties in an entirely analogous manner.  

We begin by defining random subrings of rationals based on Martin-L\"{o}f random sequences.  
First, one observes that for every subring $A$ of $(\mathbb{Q},+,\cdot)$, there is a set $P$ of primes
such that $A$ consists of all fractions $\displaystyle\frac{q}{r}$ with $q$ an integer and $r$ a product of prime powers $p_{i_0}^{n_{i_0}}$, $p_{i_k}^{n_{i_k}}$ for some $p_{i_0},\ldots,p_{i_k} \in P$.\footnote{To see this, suppose $A$
is non-trivial (otherwise the statement is immediate); then $A$ must contain $0$ as well as $1$, and therefore
by induction $A$ contains all integers.  Let $P$ be the set of primes $p$ such that $p$ has a multiplicative inverse in $A$.
Then for all $p_{i_0},\ldots,p_{i_k} \in P$ and all integers $q,n_{i_0},\ldots,n_{i_k}$, one has $qp_{i_0}^{-n_{i_0}}\ldots
p_{i_k}^{-n_{i_k}} \in A$.  Conversely, let $p$ and $q$ be relatively prime integers
such that $q > 0$ and $\frac{p}{q} \in A$.  Let $x$ and $y$ be integers with $xp+yq = 1$;
then $\frac{1}{q} = \frac{xp+yq}{q} = \frac{xp}{q} + y \in A$.  Thus for every prime factor
$r$ of $q$, $\frac{1}{r} \in A$ and so $r \in P$.}  Let $R$ be any Martin-L\"{o}f random sequence that is
Turing reducible to $K$, and let $N_R$ be the subring of $(\mathbb{Q},+,\cdot)$ such that for all $i$, $p_i^{-1} \in N_R$ iff
$R(i) = 1$.  By the preceding observation, $N_R$ consists of all fractions $\displaystyle\frac{p}{q}$ such that $p$
is any integer and $q$ is any product of prime powers $p_{i_0}^{n_{i_0}},\ldots,p_{i_k}^{n_{i_k}}$ with
$R(i_j) = 1$ for all $j \in \{0,\ldots,k\}$ and $n_{i_0},\ldots,n_{i_k} \geq 0$. 
By analogy to the definition of a generating sequence for $G_R$, a \emph{generating sequence for $N_R$}
is any infinite sequence $(b_i)_{i<\omega}$ such that $\spn{b_i\mid i < \omega} = N_R$. 
All the earlier definitions that applied to $G_R$ will be adapted, mutatis mutandis, to the subring $N_R$. 

\begin{theorem}\label{thm:randomsubringcoreequfinite}
Let $R \leq_T K$ be Martin-Löf random w.r.t the Lebesgue measure on $2^\omega$.
Then there is a generating sequence $\beta$ for $N_R$ such that
\begin{enumerate}[label=(\roman*)]
\item equality with respect to $\beta$ is co-r.e.;
\item for any finitely generated subgroup $F$ of $N_R$, the set of representations of $F$ with respect to
$\beta$ is co-r.e.;
\item the class of all sets of representations of finitely generated subgroups of $N_R$ with respect to $\beta$ is explanatorily
learnable using co-r.e.\ indices.  
\end{enumerate}  
\end{theorem}        

\begin{proof}
We follow the construction of $\beta$ in the proof of Theorem \ref{thm:genseqcoreequality} with a few modifications.
Fix some prime $p$ such that $p^{-1} \in N_R$; the Martin-L\"{o}f randomness of $R$ implies that such a $p$
exists.  As before, $R_s$ denotes the $s$-th approximation of $R$; without loss of generality, assume that 
for all $t > s$, $R(t) = 0$.  The $(i+1)$-st term of $\beta$ will be denoted by $\beta(i)$, while the $s$-th approximation
of $\beta$ will be denoted by $\beta^s$.  For any $\alpha \in \mathbb{Q}^{<\omega}$, the $(i+1)$-st term of $\alpha$
will be denoted by $\alpha(i)$.  The construction of $\beta$ proceeds in stages.  For any sequence $\gamma$, 
$i < |\gamma|$ and $r \in \mathbb{Q}$, $\gamma[i \rightarrow r]$ denotes
the sequence obtained from $\gamma$ by replacing its $(i+1)$-st term with $r$. 
\begin{enumerate}
\item Stage $0$.  Set $\beta^0 = (1,p^{-1})$.
\item Stage $s+1$. 
\begin{enumerate} 
\item  Compute $R_{s+1},R_{s+2},R_{s+3},\ldots$ in succession until the least $s' \geq s+1$ is found such that for some 
minimum $i \leq s'$, $R_{s'}(i) = 1$ and $p_{i}^{-1}$ is not a term of $\beta^s$.
(The Martin-L\"{o}f randomness of $R$ implies that such $s'$ and $i$ exist.)
Set $\beta^{s+1} \leftarrow \beta^s \circ (p_i^{-1})$.
For each $p_j$ such that $j = i$ or $\beta^s$ contains a term equal to $p_j^{-1}$, and
for $m=1$ to $m=s+1$, if $p_j^{-m}$ is not a term of $\beta^s$, set $\beta^{s+1} \leftarrow \beta^{s+1} \circ (p_j^{-m})$.
(This step ensures that for all $p_j^{-1} \in N_R$, $\beta$ eventually contains all terms of the shape $p_j^{-m}$,
 where $m \geq 1$.) 
Then go to Step 2.b. 
\item Check for every $j \leq |\beta^{s+1}|-1$ whether the $(j+1)$-st term of $\beta^{s+1}$ equals $p_{\ell}^{-m}$ 
for some $\ell$ and $m \geq 1$ such that $R_{s+1}(\ell) = 0$.  Suppose there is a least such $j$, say $j'$.  Then the 
$(j'+1)$-st term of $\beta^{s+1}$ is replaced with $p^{-n}$ for some $n > s+1$ that is large enough so that 
$n > 2n'+s+1$ for all $p^{-n'}$ occurring
in $\beta^{s+1}$ and all inequalities over the range $\{-s-1,\ldots,0,\ldots,s+1\}$ with respect to
$\beta^{s+1}$ are preserved; in other words, for all pairs $(\sigma,\sigma') \in 
\{-s-1,\ldots,0,\ldots,s+1\}^{< |\beta^{s+1}|+1} \times \{-s-1,\ldots,0,\ldots,s+1\}^{< |\beta^{s+1}|+1}$ 
such that $\sigma \cdot \beta^{s+1}_{|\sigma|-1} 
\neq \sigma' \cdot \beta^{s+1}_{|\sigma'|-1}$, one also has the relation $\sigma \cdot \beta'_{|\sigma|-1} \neq \sigma' \cdot 
\beta'_{|\sigma'|-1}$, where $\beta' = \beta^{s+1}[j' \rightarrow p^{-n}]$.  
Set $\beta^{s+1} \leftarrow \beta^{s+1}[j' \rightarrow p^{-n}]$, then go to Step 2.c.
\item Repeat Step 2.b until no term of $\beta^{s+1}$
is equal to $p_{\ell}^{-m}$ for some $m \geq 1$ and $\ell$ with $R_{s+1}(\ell) = 0$, then
go to Stage $s+2$. 
\end{enumerate}   
\end{enumerate}  
Set $\beta = \lim_{s \rightarrow \infty} \beta^s$.  Then by Step 2.a, for every $p_j^{-1} \in N_R$ and $n \geq 1$,
$\beta$ contains a term equal to $p_j^{-n}$.  Since, as was observed earlier, every $x \in N_R$ is of the shape
$qp_{i_0}^{-n_{i_0}}\ldots p_{i_k}^{-n_{i_k}}$ for some $p_{i_0}^{-1},\ldots,p_{i_k}^{-1} \in N_R$
and $q \in \Z$, $\beta$ is a generating sequence for $N_R$.  It remains to verify that $\beta$ satisfies (i), (ii) and (iii).

(i) It suffices to show that $\NE_{\beta} := \{(\sigma,\sigma')\in \Z^{<\omega} \times \Z^{<\omega}\mid
\sigma \cdot \beta_{|\sigma|-1} \neq \sigma' \cdot \beta_{|\sigma'|-1}\}$ is equal to
the r.e.\ set $\bigcup_{s < \omega} \NE^s_{\beta}$, where $\NE^s_{\beta} := \{(\tau,\tau') \in 
\{-s,\ldots,0,\ldots,s\} \times \{-s,\ldots,0,\ldots,s\}\mid |\tau|,|\tau'| \leq |\beta^s| \wedge 
\tau \cdot \beta^s_{|\tau|-1} \neq \tau' \cdot \beta^s_{|\tau'|-1}\}$. 
Suppose $(\sigma,\sigma') \in \NE_{\beta}$.  Since $R$ is limiting-recursive, there is an 
$s_0 \geq \max(\ran(\sigma) \cup \ran(\sigma'))$ large enough so that whenever $t \geq s_0$, 
$|\beta^t| \geq \max(\{|\sigma|,|\sigma'|\})$ and the value of $\beta^t(i)$ for every $i$ in the domain
of $\sigma$ or $\sigma'$ has stabilised, i.e.\ $\beta^t(i) = \beta(i)$ for all $i < \max(\{|\sigma|,|\sigma'|\})$.
It follows that $(\sigma,\sigma') \in \NE^{s_0}_{\beta}$.  

Now suppose $(\tau,\tau') \in \NE^{s_1}_{\beta}$
for some $s_1 \in \N$, so that $\tau \cdot \beta^{s_1}_{|\tau|-1} \neq \tau' \cdot
\beta^{s_1}_{|\tau|-1}$.  By Steps 2.b and 2.c in the construction of $\beta$, one has
$\tau \cdot \beta^t_{|\tau|-1} \neq \tau' \cdot \beta^t_{|\tau'|-1}$ for all $t \geq s_1$.
In particular, if $s_2 \geq s_1$ is the least number such that for all $t' \geq s_2$,
$\beta^{t'}(i) = \beta(i)$ whenever $i < \max(\{|\sigma|,|\sigma'|\}$, then
$\tau \cdot \beta^{s_2}_{|\tau|-1} \neq \tau' \cdot \beta^{s_2}_{|\tau'|-1}$, which is equivalent
to $\tau \cdot \beta_{|\tau|-1} \neq \tau' \cdot \beta_{|\tau'|-1}$.  Therefore
$(\tau,\tau') \in \NE_{\beta}$.  

(ii) Let $F$ be any finitely generated subgroup of $N_R$ that is generated by $\displaystyle\frac{r}{r'}$
for some relatively prime integers $r$ and $r'$ with $r' > 0$.  Note that every term of $\beta$
is equal to $1$ or is of the shape $p_i^{-m'}$ for some $i$ and $m' \geq 1$, and that $\beta$ is a one-one
sequence.
Hence there is a least $s_0$ such that for all $j \geq s_0$, $\beta(t) \notin F$;
furthermore, the Martin-L\"{o}f randomness of $R$ implies that $s_0$ can be chosen
so that $\beta_{s_0}$ contains a term of the shape $p^{-n'}$ for some $n' \in \N$
with $p^{n'} \nmid r'$.
Fix some $s_1 \geq s_0$ such that for all $t' \geq s_1$, $|\beta^{t'}| \geq s_0+1$
and $\beta^{t'}(i) = \beta(i)$ whenever $i < s_0$.  We claim that the complement of the set of 
representations of $F$ with respect to $\beta$, denoted by $\overline{F}_{\beta}$, is equal to the r.e.\
set
$$
\bigcup_{t \geq s_1}\{ \sigma \in \{-t,\ldots,0,\ldots,t\} \times \{-t,\ldots,0,\ldots,t\}\mid 
|\sigma| \leq |\beta^t| \wedge \sigma \cdot \beta^t_{|\sigma|-1} \notin F\}.
$$
That the latter set contains $\overline{F}_{\beta}$ follows from the fact that for all 
$\sigma \in \overline{F}_{\beta}$, there is an $s_2$ such that whenever $t \geq s_2$,
$|\beta^t| \geq |\sigma|$ and $\beta^t(i) = \beta(i)$ for all $i < |\sigma|$.

Now consider any $\sigma \in \{-t,\ldots,0,\ldots,t\} \times \{-t,\ldots,0,\ldots,t\}$ and 
$t \geq s_1$ such that $|\sigma| \leq |\beta^t|$
and $\sigma \cdot \beta^t_{|\sigma|-1} \notin F$.  It will be shown that $\sigma \in \overline{F}_{\beta}$.
Consider all $i_0,\ldots,i_k \in \{0,\ldots,|\sigma|-1\}$ 
such that for each $j \in \{0,\ldots,k\}$, there is some $t' > t$ with
$\beta^{t'}(i_j) \neq \beta^t(i_j)$.  It may be assumed without loss of generality
that $\sigma(i_j) \neq 0$ for all $j \in \{0,\ldots,k\}$ (for if $\sigma(i_j) = 0$,
then any difference between the value of $\beta^t(i)$ and $\beta^{t'}(i)$ would have no 
effect on whether $\sigma \in \overline{F}_{\beta}$).  By Steps 2.b and 2.c in the
construction of $\beta$, there are $n_0,\ldots,n_k > t$ with $n_{i+1} > 2 n_i+t$ for all 
$i \in \{0,\ldots,k-1\}$ such that $\{\beta(i_j)\mid 0 \leq j \leq k\} =
\{p^{-n_i}\mid 0 \leq i \leq k\}$.  Without loss of generality, assume that
$\beta(i_j) = p^{-n_j}$ for all $j \in \{0,\ldots,k\}$.
Then $\sum_{j=0}^k \sigma(i_j) \beta(i_j) = \sum_{j=0}^k \sigma(i_j)p^{-n_j}$. 
Since $0 < |\sigma(i_j)| \leq t$, $n_j > t$ and $n_{j+1} > 2n_j+t$, there
is an $n'' \in \N$ such that $\sum_{j=0}^k \sigma(i_j)p^{-n_j}$ equals $qp^{-n''}$ for 
some $q \in \Z$ that is coprime to $p$ and $n'' > 2n'''$ for the largest $n'''$ such that $p^{-n'''}$
is a term of $\beta^t$ (such a term exists by the choice of $s_0$).  Thus $\sigma \cdot \beta_{|\sigma|-1}$
is of the shape $\displaystyle\frac{r''}{r'''}$ for some relatively prime integers
$r''$ and $r'''$ such that $p^{n''}$ divides $r'''$.  Since $\beta^t$ contains a term
$p^{-n'}$ such that $p^{n'} \nmid r'$ (where, as stated at the beginning of the proof,
$r$ and $r'$ are relatively prime integers with $r' > 0$ such that $\displaystyle\frac{r}{r'}$ is a 
generator of $F$), it follows that $p^{n''} \nmid r'$ and therefore $\sigma \cdot \beta_{|\sigma|-1}
\notin F$.  Consequently, $\sigma \in \overline{F}_{\beta}$. 

(iii)  We first observe the following.
Suppose $T$ is any text for $F_{\beta}$, where $F = \spn{rq^{-1}}$
for some relatively prime integers $r,q$ with $q > 0$.  
\begin{enumerate}[label=(\Roman*)]
\item Since $\beta(0) = 1$, the value of $r$ can be determined in the limit from $T$ by 
taking the greatest common divisor of all elements $w$ such that $(w) \in \ran(T)$. 
\item For every prime power factor $p_i^{n_i}$ of $q$, there is some $i'$ such that 
$\beta(i') = p_i^{-n_i}$, and therefore $I_{i'+1}(r) \in \ran(T)$.
\item There is a least $j$ such that for all $j' > j$, $\beta(j')$ is of the shape
$p_{\ell}^{-n'}$ for some prime $p_{\ell}$ and $n' \geq 1$ with
$rp_{\ell}^{-n'} \notin F$; in particular, $I_{j'+1}(r) \notin \ran(T)$. 
\item By (ii), there is a minimum $s_0$ such that 
$\overline{F}_{\beta} = \bigcup_{t \geq s_0}\{ \sigma \in \{-t,\ldots,0,\ldots,t\} \times \{-t,\ldots,0,$ $\ldots,t\}\mid 
|\sigma| \leq |\beta^t| \wedge \sigma \cdot \beta^t_{|\sigma|-1} \notin F\}$.
\end{enumerate}
Define a learner $M$ as follows.  On input $\delta$, $M$ determines the
greatest common divisor $d$ of the set of all $w$ such that $(w) \in \ran(\delta)$
(if no such $w$ exists, then $M$ just sets $d=0$).
Next, $M$ identifies all $j_0,\ldots,j_{\ell}$ such that $I_{j_k+1}(d) \in \ran(\delta)$
for all $k \in \{0,\ldots,\ell\}$, and it approximates $\beta(j_k)$ by determining
the least $t' \geq |\delta|$ such that $|\beta^{t'}| \geq j_k+1$ and setting
the approximation to be $\beta^{t'}(j_k)$.  
For each $k \in \{0,\ldots,\ell\}$, let $p_{i_k}^{-n_{i_k}}$ be the current approximation
of $\beta(j_k)$, where $n_{i_k} \geq 0$.  $M$ then takes $dp_{h_0}^{-m_0}\ldots p_{h_{l'}}^{-m_{\ell'}}$ 
to be its current guess for a generator of $F$,
where $\{p_{h_0},\ldots,p_{h_{\ell'}}\} = \{p_{i_0},\ldots,p_{i_\ell}\}$ and
for each $k \in \{0,\ldots,\ell'\}$, $m_k$ is the largest number $e'$ such that $p_{h_k}^{-e'} \in 
\{p_{i_0}^{-n_{i_0}},\ldots,p_{i_{\ell}}^{-n_{i_{\ell}}}\}$.   
Having determined a guess for $F$, $M$ finds the least $s_0 \leq |\delta|$ such that the
$|\delta|$-th approximation of
$$
G_{s_0} := \bigcup_{t \geq s_0}\{ \sigma \in \{-t,\ldots,0,\ldots,t\} \times \{-t,\ldots,0,\ldots,t\}\mid 
|\sigma| \leq |\beta^t| \wedge \sigma \cdot \beta^t_{|\sigma|-1} \notin F\},
$$
denoted by $G_{s_0,|\delta|}$, satisfies $G_{s_0,|\delta|} \cap \ran(\delta) = \emptyset$.
If no such $s_0$ exists, then $M$ outputs a co-r.e.\ index for $\emptyset$; otherwise, $M$ outputs a
co-r.e.\ index for $\Z^{<\omega} \sm G_{s_0}$.  By (I), (II) and (III), $M$ on $T$ will correctly identify a generator for 
$F$ in the limit.  Furthermore, defining $s_0$ as in (IV), if $s_0 \geq 1$, then $\ran(T)$ contains
some element in $\bigcup_{t \geq s_0-1}\{ \sigma \in \{-t,\ldots,0,\ldots,t\} \times \{-t,\ldots,0,$ $\ldots,t\}\mid 
|\sigma| \leq |\beta^t| \wedge \sigma \cdot \beta^t_{|\sigma|-1} \notin F\}$; thus, for all $s' < s_0$,
$M$ will reject $\Z^{<\omega} \sm G_{s'}$ and conjecture $\Z^{<\omega} \sm G_{s_0}$ as the correct hypothesis 
in the limit.
\end{proof}

\begin{remark}
The explanatory learner $M$ in the proof of (iii) of Theorem \ref{thm:randomsubringcoreequfinite} is also \emph{conservative}
in the sense that for any two text initial segments $T[n_1]$ and $T[n_2]$ for any $F_{\beta}$, where $n_1 < n_2$,
$M(T[n_1]) \neq M(T[n_2])$ only if $\ran(T[n_2]) \not\subseteq \overline{W}_{M(T[n_1])}$ (we assume
that $M$'s hypothesis space is some fixed numbering $\overline{W}_0,\overline{W}_1,\overline{W}_2,\ldots$ of co-r.e.\ 
subsets of $Z^{<\omega}$).    
\end{remark}

\noindent
We recall that for any prime $p$, a \emph{Pr\"{u}fer $p$-group} (denoted by $\Z(p^{\infty})$) may be defined as the quotient 
of the group of all rational numbers whose denominator is a power of $p$ by $\Z$.  Regarding $\Z(p^{\infty})$ as a subgroup
of $(\mathbb{Q}/\Z,+)$, we define a \emph{random join of Pr\"{u}fer groups} based on any given Martin-L\"{o}f random
sequence $R$ as follows.  As before, suppose $R$ is Martin-L\"{o}f random and is Turing reducible to $K$.  Then
the subgroup $P_R$ is defined to be the join of all $\Z(p_{i_0}^{\infty}),\Z(p_{i_1}^{\infty}),\Z(p_{i_2}^{\infty}),\ldots$
such that for all $j \in \N$, the $i_j$-th bit of $R$ is $1$.  In other words, $P_R$ consists of every fraction (modulo $1$) 
whose denominator is a product of finitely many powers of primes belonging to $\{p_{i_0},p_{i_1},p_{i_2},\ldots\}$.
The next result is the analogue of Theorem \ref{thm:randomsubringcoreequfinite} for $P_R$; the proof is entirely
similar to that of Theorem \ref{thm:randomsubringcoreequfinite}.

\begin{theorem}\label{thm:randomjoinprufercoreequfinite}
Let $R \leq_T K$ be Martin-Löf random w.r.t the Lebesgue measure on $2^\omega$.
Then there is a generating sequence $\beta$ for $P_R$ such that
\begin{enumerate}[label=(\roman*)]
\item equality with respect to $\beta$ is co-r.e.;
\item for any finitely generated subgroup $F$ of $P_R$, the set of representations of $F$ with respect to
$\beta$ is co-r.e.;
\item the class of all sets of representations of finitely generated subgroups of $P_R$ with respect to $\beta$ is explanatorily
learnable using co-r.e.\ indices.  
\end{enumerate}  
\end{theorem}     

\noindent
By adapting the proofs of Theorems \ref{thm:reequalitymod1}, \ref{thm:finitelygeneratedmod1bc} and \ref{thm:finitegeneratednotex}, one obtains an almost ``symmetrical'' version of Theorem 
\ref{thm:randomjoinprufercoreequfinite} for $P_R$.

\begin{theorem}\label{thm:randomjoinprufercoreequfinitecor}
Let $R \leq_T K$ be Martin-Löf random w.r.t the Lebesgue measure on $2^\omega$.
Then there is a generating sequence $\beta$ for $P_R$ such that
\begin{enumerate}[label=(\roman*)]
\item equality with respect to $\beta$ is r.e.;
\item for any finitely generated subgroup $F$ of $P_R$, the set of representations of $F$ with respect to
$\beta$ is r.e.;
\item the class of all sets of representations of finitely generated subgroups of $P_R$ with respect to $\beta$ is
$\Bc$-learnable but not $\Ex$-learnable.  
\end{enumerate}  
\end{theorem}

\section{Conclusion and Possible Future Research}

This paper introduced a method of constructing random subgroups of rationals, 
whereby Martin-L\"{o}f random binary sequences are directly encoded into the
generators of the group.  It was shown that if the Martin-L\"{o}f random
sequence associated to a randomly generated subgroup $G$ is limit-recursive,
then one can build a generating sequence $\beta$ for $G$ such that
the word problem for $G$ is co-r.e.\ with respect to $\beta$, as well as another
generating sequence $\beta'$ such that the word problem for $G/\Z$
with respect to $\beta'$ is r.e.  We also showed that every non-trivial finitely generated
subgroup of $G$ has an r.e.\ representation with respect to a suitably chosen
generating sequence for $G$; moreover, the class of all such r.e.\ representations
is behaviourally correctly learnable but never explanatorily learnable.

A question deserving further attention is the extent to which the choice of the generating sequence for a randomly generated subgroup $G$ of rationals influences the learnability of its finitely generated subgroups; in particular, is there a generating sequence $\beta$ for $G$ such that every
non-trivial finitely generated subgroup of $G$ has an r.e.\ representation with respect to $\beta$ and the class of all such representations 
with respect to $\beta$ is not even behaviourally correctly learnable?
We also did not extend the definition of algorithmic randomness to
\emph{all} Abelian groups; we suspect that such a general definition might be 
out of reach of current methods due to the fact that the isomorphism types of even rank $2$
groups (subgroups of $(\mathbb{Q}^2,+)$) are still unknown.

\end{document}